\documentclass[11pt]{article}

\usepackage{amsmath,amsthm,amssymb,amsfonts,verbatim}
\usepackage{latexsym}
\usepackage{fullpage}
\usepackage{graphicx}
\usepackage{float}
\usepackage{subfig}
\usepackage{algorithm}
\usepackage{algorithmic}
\usepackage[dvipsnames,usenames]{color}
\usepackage[hyperindex,pdftex]{hyperref}
\usepackage{times}
\usepackage{caption}
\usepackage{todonotes}

 \providecommand{\F}{\mathbb{F}}

\date{}

\newtheorem{lemma}{Lemma}
\newtheorem{theorem}{Theorem}
\newtheorem{cor}[lemma]{Corollary}
\newtheorem{prop}[lemma]{Proposition}

\newtheorem{defn}{Definition}

\newtheorem{rmk}{Remark}

\def \mC {\mathcal{C}}

\def \mB {\mathcal{B}}
\def \mC {\mathcal{C}}

\def \mL {\mathcal{L}}

\def \mS {\mathcal{S}}

\def \Xi {{X^{[i]}}}

\newcommand{\Ge}{\epsilon}

\def \wt {{\rm wt}}
\def \mS {{\mathcal S}}

\begin{document}

\onecolumn

\title{On List Decoding of Insertion and Deletion Errors}

\author{Shu Liu\thanks{Shu Liu is with the National Key Laboratory of Science and Technology on Communications, University of Electronic Science and Technology of China, Chengdu 611731, China, and also with the Division of Mathematical Sciences, School of Physical and Mathematical Sciences, Nanyang Technological University, Singapore 637371 (email: shuliu@uestc.edu.cn).},
Ivan Tjuawinata\thanks{Ivan Tjuawinata is with the Strategic Centre for Research on Privacy-Preserving Technologies and Systems, Nanyang Technological University, Singapore 637553 (ivan.tjuawinata@ntu.edu.sg).}
~and Chaoping Xing
\thanks{ Chaoping Xing is with the Division of Mathematical Sciences, School of Physical and Mathematical Sciences, Nanyang Technological University,  Singapore 637371 (email: xingcp@ntu.edu.sg).}}
\maketitle
\begin{abstract}
Insertion and deletion (Insdel for short) errors are synchronization errors in communication systems caused by the loss of positional information of the message. Since the work by Guruswami and Wang \cite{VGCW} that studied list decoding of binary codes with deletion errors only, there have been some further investigations on the list decoding of  insertion codes, deletion codes and insdel codes. However, unlike classical Hamming metric or even rank-metric, there are still many unsolved problems on  list decoding of insdel codes.

The purpose of the current paper is to move toward complete or partial solutions for  some of these problems.
Our contributions mainly consist of two parts. Firstly, we analyse the list decodability of random insdel codes. We show that  list decoding of random  insdel codes surpasses the Singleton bound when there are more insertion errors than deletion errors and the alphabet size is sufficiently large. We also find that our results improve some previous findings in~\cite{HSS} and \cite{VGCW}. Furthermore, our results reveal the existence of an insdel code that can be list decoded against insdel errors beyond its minimum insdel distance while still having polynomial list size. This provides a more complete picture on the list decodability of insdel codes when both insertion and deletion errors happen. Secondly, we construct a family of explicit insdel codes with efficient list decoding algorithm. As a result, we derive a Zyablov-type bound for insdel errors. Recently, after our results appeared, Guruswami et al. \cite{GHS2019} provided a complete solution for another open problem on list decoding of insdel codes. In contrast to the problems we considered, they provided a region containing all possible insertion and deletion errors that are still list decodable by some $q$-ary insdel codes of non-zero rate.  More specifically, for a fixed number of insertion and deletion errors, while our paper focuses on maximizing the rate of a code that is list decodable against that amount of insertion and deletion errors, Guruswami et al. focuses on finding out the existence of a code with asymptotically non-zero rate which is list decodable against this amount of insertion and deletion errors.

\end{abstract}

\section{Introduction}
Insertion and deletion (Insdel for short) errors are synchronization errors~\cite{stoc2017},~\cite{HSS} in communication systems caused by the loss of positional information of the message. They have recently attracted many attention due to their applicabilities in many interesting fields such as DNA storage and DNA analysis~\cite{SJFFH},  \cite{RW2005}, race-track memory error correction \cite{CK2017} and language processing \cite{BM2000}, \cite{Och03}.

The study of codes with insertion and deletion errors was pioneered by Levenshtein, Varshamov and Tenengolts in the 1960s~\cite{VT65},\cite{1965},\cite{1967} and \cite{1984}. This study was then further developed by Brakensiek, Guruswami and Zbarsky \cite{JB}. There have also been different directions for the study of insdel codes such as the study of some special forms of the insdel errors \cite{Schoeny17}, \cite{CS14}, \cite{Lenz2017} and \cite{Mitz2008} as well as their relations with Weyl groups
\cite{MH2018}.

\subsection*{Previous results}
\noindent Guruswami and Wang \cite{VGCW} studied list decoding of binary codes with deletion errors only. They provided the amount of deletion errors that can be tolerated by binary codes. In addition, they explicitly constructed binary codes with decoding radius close to $\frac 12$ for deletion errors only.
Wachter-Zeh~\cite{Antonia} firstly considered the list decoding of insdel codes and provided a Johnson-type upper bound on list size in terms of minimum insdel distance of a given code in 2017. Hayashi and Yasunaga~\cite{Japan} provided some amendments on the result in~\cite{Antonia} and derived a Johnson-type upper bound which is only meaningful when insertion occurs.
 Based on the indexing scheme and concatenated codes, they further provided efficient encoding and decoding algorithms by concatenating an inner code achieving this Johnson-type bound and an outer list-recoverable Reed-Solomon code achieving the classical Johnson bound. In 2018, Haeupler, Shahrasbi and Sudan~\cite{HSS} constructed a family of list decodable insdel codes through the use of synchronization strings with larger list decoding radius (beyond Johnson-type upper bound) for sufficiently large alphabet size and designed its efficient list decoding algorithm. Furthermore, instead of insdel errors, they derived some upper bounds on list decodability  for insertion  or deletion errors only. Lastly, they considered the list decodability of random codes with insertion or deletion errors only. Their results reveal that there is a gap between the upper bound on list decodability of insertion (or deletion) codes and list decodability of a random insertion (or deletion) code. Haeupler{, Rubinstein and Shahrasbi}~\cite{HRS2018} introduced probabilistic fast-decodable indexing schemes for insdel distance which reduces the computing complexity of the list decoding algorithm in~\cite{HSS}.

Previous findings that we have discussed above leave several problems: (i) what is the list decodability of a random insdel code? (ii) are there some reasonable upper bounds on list  decoding radius of insdel codes in terms of rate? (iii) is there a Zyablov-type bound for insdel codes for small alphabet size $q$?

\subsection*{Our results}
\noindent In this paper, we focus on the list decoding of insdel codes. Our results are mainly divided into two parts. 
Firstly, we analyse the list decodability of random insdel codes. Interestingly, the list decodability of random insdel codes surpasses the Singleton bound when there are more insertion errors than deletion errors with the alphabet size is sufficiently large. This characteristic is not found in codes of many other metrics such as Hamming metric~\cite{VS05}, rank-metric~\cite{Ding2015}, cover-metric~\cite{Sliu2018} and symbol-pair metric~\cite{symbol2018}. Another phenomenon that can be observed from the list decodability of random insdel codes is the existence of insdel codes that can be list decoded against insdel errors beyond its minimum distance when the alphabet size is sufficiently large. This does not happen for other metrics. This is also the first work to investigate the list decodability of random insdel codes with respect to the total number of errors where the errors does not need to only be insertion or only be deletion.


Secondly, we construct a family of $q$-ary insdel codes which can be efficiently list decoded. Since the construction of explicit insdel codes for sufficiently large $q$ has been discussed in~\cite{HSS}, our construction focuses on smaller $q,$ even when $q=2.$ As a result, we derive a Zyablov-type bound.


\begin{rmk}
{After our results appeared, Guruswami et al.~\cite{GHS2019} studied a different open question from the one considered in this paper. In their paper, they studied the existence of an asymptotically non-zero rate code which is list decodable against a fixed amount of insertion and deletion errors.}
\end{rmk}

\subsection*{Comparisons}

\noindent Although the authors of \cite{Antonia} and \cite{Japan} studied list decodability of insdel codes, they mainly focused on {the} relation of distance and list decoding radius. As a random insdel code achieves the Gilbert-Varshamov bound, we can derive list decodability of a random code by plugging the minimum distance to the Johnson bound.  In this sense, they derived a list decodability of a random code. Our approach is different, we provide list decodability of a random code directly. It shows that our bound is better than the bound derived from Johnson bound of a random code given in \cite{Japan}. As other investigations \cite{VGCW},\cite{HSS} considered insertion or deletion only,  we have to degenerate our bounds on insdel errors to the insertion only case and deletion only case when comparing them with the previous results. When degenerating our result on list decodability of random binary codes to deletion only, we obtain the same result given in \cite{VGCW}. Furthermore, when list decodability of random insdel codes is degenerated  to insertion errors only, our result is better than those in \cite{HSS} for some parameter regimes. When   list decodability of random insdel codes is degenerated to deletion errors only, we get the same result as in \cite{HSS}.

Lastly for explicit construction, our Zyablov-type bound is better than Johnson-type bound given in \cite{Japan}.
When degenerating our explicit insdel codes to binary code with deletion only and decoding radius close to $\frac12$, we get the same result as in \cite{VGCW}.

 \subsection*{Our techniques}

\noindent To obtain our result on the list decodability of random insdel codes, the key part is to estimate the size of an insdel ball. This is much more complicated than classical Hamming  or rank-metrics. We develop some tricks to get tighter bounds on size of insdel balls.

Firstly, we only estimate the number of vectors in the insdel ball with the same length as the code length. This directly eliminates all the other elements of the insdel ball with inappropriate lengths. Secondly, due to the minimality requirement of insdel distance and the commutativity of insertion and deletion operations up to some repositionings, to enumerate these vectors, we separate to two phases; insertion phase and deletion phase where we use existing estimates on both phases when only one of the operations occurs. In contrast to the insertion sphere size which can be calculated exactly, we only have an upper bound and a lower bound for the deletion sphere size which are not asymptotically tight and depend on the number of runs of the centre. To improve the upper bound on the insdel ball size, we classify the possible centres to several cases based on the number of runs that they have. In each case, an upper bound is then derived to apply for any possible centres in the given case so the bounds are applicable to all centres in that case. Despite the bound being derived to be simultaneously applicable to all centre in the same case, it actually leads to asymptotically a stricter upper bound in all cases. Having this upper bound on the estimate of insdel ball size, it is then used in the analysis of the list decodability of random insdel codes.

%
More specifically, the list decoding radius of a random insdel codes of rate $R$ can be analysed in two steps. Firstly, we compute the probability that a random code of rate $R$ is list decodable up to normalized list decoding radius $\tau.$ Having this probability, we derive a restriction of $R$ and $\tau$ to make this probability negligibly close to $0.$

As for  our explicit construction, to increase the rate of our concatenated code, we reduced the indexing scheme size. Due to this reduction, after the inner decoding, we can no longer directly identify the correct position of each element in the list with respect to the outer codeword. We took an additional step to optimize the classification of the possible position lists. Using this technique, fixing the list decoding radius, we obtained a code with higher rate. This is true even compared to a concatenated code in~\cite{Japan} with outer code and inner code being the ones used in our construction. Furthermore, compared to the construction in~\cite{Japan}, instead of having a separate requirement on the number of insertion errors and deletion errors, our code only bounds the  combined number of insertion and deletion errors, allowing our code to list decode a wider range of insertion and deletion errors.

\subsection*{Organization}
This paper is organized as follows. In Section \ref{sec:prelim}, we introduce definitions of insdel codes and some preliminaries on list decoding. Section \ref{BallBound} contains the upper bound on the number of fixed length words in an insdel ball.  Section \ref{sec:random} is dedicated to the analysis of the list decodability of random insdel codes. Lastly, the construction and decoding algorithm of our list decodable insdel codes are provided in Section \ref{ConsSection}.
\section{Preliminaries}\label{sec:prelim}
{
Let $\Sigma_q$ be a finite alphabet of size $q$ and $\Sigma_q^n$ be the set of all vectors of length $n$ over $\Sigma_q.$ For any positive real number $i,$ we denote by $[i]$ the set of integers $\{1,\cdots, \lfloor i\rfloor\}.$



\begin{defn}[Insdel distance]
The insdel distance $d(\mathbf{a},\mathbf{b})$ between two words ${\bf a}\in\Sigma_q^{n_1}$ and ${\bf b}\in\Sigma_q^{n_2}$ (not necessarily of the same length) is the minimum number of insertions and deletions which is needed to transform $\bf a$ into $\bf b.$
\end{defn}

{Note that for two vectors $\bf a$ of length $n_1$ and $\bf b$ of length $n_2,$ $d(\bf a, \bf b)$ is at least $|n_1-n_2|$ and is at most $n_1+n_2.$
The minimum insdel distance of a code $\mC\subseteq\Sigma_q^n$ is defined as
$d(\mC)=\mathop{\min}\limits_{\mathbf{a}, \mathbf{b}\in\mC, \mathbf{a}\neq\mathbf{b}}\{d(\bf a, \bf b)\}.
$
A code over $\Sigma_q$ of length $n$ with size $M$ and minimum insdel distance $d$ is called an $(n, M, d)_q$-insdel codes. Similar to classical Hamming metric codes, we can define the rate and the relative insdel distance of an $(n, M, d)_q$-insdel code $\mC$ by $R(\mC)=\frac{\log_q{ M}}{n}$ and $\delta(\mC)=\frac{d}{2n}.$
The relative insdel distance is normalized by $2n$ instead of $n$ since the insdel distance between two words in $\Sigma_q^n$ takes a nonnegative integer value up to $2n.$

The minimum insdel distance is one of the important parameters for an insdel code. So, it is desirable to keep minimum insdel distance $d$ as large as possible for an insdel code with fixed length $n.$ It has been shown~\cite{stoc2017} that an $(n, M, d)_q$-insdel code $\mC$ must obey the following version of the Singleton bound.
\begin{prop}[Singleton Bound~\cite{stoc2017}]\label{SingBound}
Let $\mathcal{C}\subseteq \Sigma_q^n$ be an $(n, M, d)_q$-insdel code of length $n$ and minimum insdel distance $0\leq d\leq 2n,$ then
$M\leq q^{n-{d}/{2}+1}.$
\end{prop}
An asymptotic way to state the Singleton bound for an insdel code $\mC$ in term of its rate and relative minimum insdel distance is
$R(\mC)+\delta(\mC)\leq 1.
$

An $[n, k, d]_q$-insdel code is a $\Sigma_q$-linear code over $\Sigma_q$ of length $n,$ dimension $k$ and minimum insdel distance $d.$ Then, we provide the definitions of an insertion (or deletion) sphere and an insdel ball.

\begin{defn}[Sphere] For a word ${\bf u}\in\Sigma_q^n$ and a nonnegative real number $z,$ the deletion sphere centered at ${\bf u}$ with radius $z$ is defined by
$\mathcal{S}_{\sf D}(\mathbf{u},z) =\left \{\mathbf{v}\in \Sigma_q^{n-z}: \mathbf{v} \textrm{~can~be~obtained~from~}\mathbf{u}\textrm{~by~}z\textrm{~deletions}\right\}.$
Insertion sphere, denoted by $\mathcal{S}_{\sf I}(\mathbf{u},z),$ can be defined similarly.
\end{defn}


The insdel ball, as an analogue to the Hamming metric ball, is used to count the number of words within a given insdel distance.

\begin{defn}[Insdel Ball] For a word ${\bf u}\in\Sigma_q^n$ and a nonnegative real number $z,$ the insdel ball centered at ${\bf u}$ with radius $z$ is defined by
$\mathcal{B}(\mathbf{u},z) = \left\{\mathbf{v}\in \bigcup_{i=\max\{n-z,0\}}^{n+z}\Sigma_q^i: d(\mathbf{u},\mathbf{v})\leq z\right\}.$
\end{defn}


We now proceed to the definition of list decodability of insdel codes and the entropy function.
\begin{defn}
For a real $\tau\geq 0,$ an insdel code $\mC\subseteq\Sigma_q^n$ is said to be $(\tau n,L)$-list-decodable, if for every nonnegative integer $m\in[\max(0,n-\tau n), n+\tau n]$ and every ${\bf r}\in\Sigma_q^{m},
|\mathcal{B}({\bf r},\tau n)\cap \mC |\leq L.$
\end{defn}

In this paper, although we require that the number of symbols deleted must be less than the length of the code, there is no such restriction for number of symbols inserted. So the number of deletion errors must be at most the length of the code while both the number of insertion errors and the isndel errors $\tau n$ can be any non-negative integers.

For a non-zero rate code $\mathcal{C}$ of length $n$ over an alphabet of size $q, \Sigma_q=\{0,1,\cdots, q-1\},$ in order to have polynomial number of codewords in insdel balls with any centre, the number of deletion errors must be less than $\frac{q-1}{q} n.$ This is because with $\frac{q-1}{q}n$ deletions, any codeword $\mathbf{c}\in\mathcal{C}$ can be transformed to a word of length $\frac{n}{q}$ with all entries being the symbol most frequently occurring in $\mathbf{c}.$ Similarly, the number of insertion errors must be less than $(q-1)n$ because with $(q-1)n$ insertions, any codeword $\mathbf{c}\in\mathcal{C}$ can be transformed to the word $(0,1,\cdots,q-1,0,1,\cdots,q-1,\cdots,0,1,\cdots,q-1)\in\Sigma_q^{qn}.$


\begin{defn}[$q$-ary Entropy Function] Let $q$ be an integer and $x$ be a real number such that $q\geq 2$ and $0< x< 1.$ The $q$-ary entropy function, $H_q(x)$ is defined as follows $
H_q(x)=x\log_q(q-1)-x\log_q(x)-(1-x)\log_q(1-x).$
By convention, we define $H_q(0)=H_q(1)=0.$
\end{defn}

By Proposition 3.3.2 in~\cite{Hqxx}, when $q$ is sufficiently large, $H_q(x)$ can be approximated by $x$ with arbitrarily small error.

Finally, we provide the definition of a list-recoverable code. List-recoverable codes are used in the explicit construction discussed in Section~\ref{ConsSection}. The study of list-recoverable codes was inspired by Guruswami-Sudan's list decoding algorithm for Reed-Solomon codes \cite{GS98}. Many list-recoverable codes have been constructed such as \cite{GR08}, \cite{GW11}, \cite{GX13},\cite{GX14}, \cite{HRW17}, \cite{HW15} and \cite{Kop15}. In this paper, we use an alternative definition of list-recoverable code.

\begin{defn}\label{def:LRC}
Let $0<\alpha<1$ be a real number, $\ell$ and $L$ be two positive integers. A code $\mathcal{C}\subseteq \Sigma_q^n$ is said to be $(\alpha,\ell,L)$-list-recoverable if for any given $n$ sets $S_1,\cdots, S_n\subseteq \Sigma_q$ such that $\sum_{i=1}^n |S_i|\leq \ell,$ we have
$\left|\left\{\mathbf{x}=(x_1,\cdots, x_n)\in \mathcal{C}: \left|\left\{ i\in [n]: x_i\in S_i\right\}\right|\geq \alpha n\right\}\right|\leq L.$
\end{defn}

}

\section{Analysis of insdel ball}\label{BallBound}

Note that our interest is in the number of codewords in an insdel ball and our insdel code $\mC$ is over $\Sigma_q^n.$ So to have an estimate that is independent of the actual $\mC$, this paper focuses on the set of vectors of length $n$ in the insdel ball, $\mathcal{B}(\mathbf{r},\tau n)\cap\Sigma_q^n$.

The insertion and deletion operations are commutative up to some adjustments. So, the order of the operations from one word to the other does not matter as long as the number of deletions and insertions are the same.
Let $\mathcal{C}\subseteq \Sigma_q^n$ be an insdel code and ${\bf r}\in\Sigma_q^{m}$ be the received word. We assume that an insdel error of size at most $\tau n$ occurs during transmission for some $\tau\geq 0.$ We further assume that $\gamma$ fraction of insertions and $\kappa$ fraction of deletions occurred to obtain ${\bf r},$ where $\gamma\geq 0$ and $0\leq \kappa\leq 1.$
So, we have
\begin{equation}
\mathcal{B}({\bf r},\tau n)\cap \Sigma_q^n=\bigcup_{\gamma+\kappa\leq \tau}\bigcup_{{\bf c}'\in \mathcal{S}_{\sf I}({\bf r},\kappa n)} \mathcal{S}_{\sf D}({\bf c}', \gamma n)=\bigcup_{\gamma +\kappa \leq \tau }\bigcup_{{\bf c}'\in \mathcal{S}_{\sf D}({\bf r}', \gamma n)} \mathcal{S}_{\sf I}({\bf c}', \kappa n),
\end{equation}
where $\gamma \leq \frac{\tau n - n + m}{2n}$ and $\kappa \leq \frac{\tau n + n-m}{2n}.$
Define $\gamma^* = \frac{\tau n - n + m}{2n}$ and $\kappa^* = \frac{\tau n + n-m}{2n}$.

\begin{defn}
For any sequence $\mathbf{s}$ and a non-negative integer $n\geq 0,$ we define $\mathbf{s}^n$ as follows:
\begin{equation*}
\mathbf{s}^n : =\left\{
\begin{array}{cc}
\xi & {\rm~if~}n=0,\\
(\underbrace{\mathbf{s},\mathbf{s},\cdots,\mathbf{s}}_n)& {\rm otherwise}.
\end{array}
\right.
\end{equation*}
Here $\xi$ represents the empty sequence of length $0.$ For $s\in\Sigma_q,$ to avoid confusion, we will define the repetition sequence by $(s)^n$ instead of $s^n.$
\end{defn}
\begin{defn}
For a positive integer $n,$ we define the repetition set $R_q(n)\subseteq \Sigma_q^n$ as
\begin{equation*}
R_q(n)=\{(\alpha)^n\in\Sigma_q^n: \alpha\in\Sigma_q\}.
\end{equation*}
\end{defn}



In order to obtain the tight bound of  the size of $\mathcal{B}({\bf r},\tau n)\cap \Sigma_q^n$, we consider two cases to discuss depending on the form of the received word ${\bf r}\in\Sigma_q^{m}.$

Firstly, we consider $|\mathcal{B}({\bf r},\tau n)\cap \Sigma_q^n|$ when the received word ${\bf r}\in R_q(m).$
\begin{lemma}~\label{theorem:case1}
Given a received word  ${\mathbf{r}}\in R_q(m)\subseteq \Sigma_q^{m}$ with $m\in[n-\tau n,n+\tau n]$ and $\kappa^\ast = \frac{\tau n + n-m}{2}\leq \frac{q-1}{q},$ we have
\begin{eqnarray*}
|\mathcal{B}({\bf r},\tau n)\cap \Sigma_q^n|= q^{n\left(H_q(\kappa^\ast)+O\left(\frac{\log_q(n)}{n}\right)\right)}.
\end{eqnarray*}
\end{lemma}
\begin{proof}
Without loss of generality, assume ${\bf r} = (0)^{m}.$ Note that for any $\mathbf{x}\in\mathcal{B}({\bf r},\tau n)\cap\Sigma_q^n$ with Hamming weight $w,$ all of the non-zero elements must appear in the insertion phase from $\mathbf{r}.$ {For any vector $\mathbf{v},$ denote by $\wt_{\sf H}(\mathbf{v})$ the Hamming weight of $\mathbf{v}.$} Since we can insert at most $\kappa n$ symbols, we have $\wt_{\sf H}({\bf x})=w\leq \kappa n\leq \frac{\tau n + n-m}{2}.$
The last inequality comes from the fact that $\gamma+\kappa\leq \tau$ and $\kappa n-\gamma n = n-m.$
Hence, the insertion and deletion processes can be regrouped to two main steps: adjusting the number of zeros with the appropriate number of insertion or deletion and then inserting the non-zero symbols.
So, when we enumerate the number of elements of $\mathcal{B}(\mathbf{r},\tau n)\cap \Sigma_q^n$ with weight $w,$ we first transform $(0)^{m}$ to $(0)^{n-w}$ before inserting all the $w$ non-zeros.
Enumerating all possible ${\bf x}\in\mathcal{B}(\mathbf{r},\tau n)\cap \Sigma_q^n,$ we have
\begin{eqnarray*}
|\mathcal{B}({\bf r},\tau n)\cap \Sigma_q^n|=\sum_{w=0}^{\frac{\tau n + n - m}{2}} \left|\left\{{\bf c}\in \mathcal{S}_{\sf I}((0)^{n-w},w), \wt_{\sf H}({\bf c})= w\right\}\right|
=\sum_{w=0}^{\frac{\tau n + n - m}{2}}\binom{n}{w}(q-1)^w.
\end{eqnarray*}

Since the maximum term is when $w=\frac{\tau n + n-m}{2}= \kappa^* n,$ the maximum summand in the last term is $\binom{n}{\kappa^\ast n}(q-1)^{\kappa^\ast n}= q^{n\left(H_q(\kappa^\ast)+O\left(\frac{\log_q(n)}{n}\right)\right)}.$ Hence $|\mathcal{B}(\mathbf{r},\tau n)\cap \Sigma_q^n|$ can be bounded by

\begin{equation*}
q^{n\left(H_q(\kappa^\ast)+O\left(\frac{\log_q(n)}{n}\right)\right)}\leq |\mathcal{B}(\mathbf{r},\tau n)\cap \Sigma_q^n|\leq q^{n\left(H_q(\kappa^\ast)+O\left(\frac{\log_q(n)}{n}\right)\right)+O(\log_q(n))}.
\end{equation*}

Note that when $n$ is sufficiently large, the two bounds are the same, which is $q^{n\left(H_q(\kappa^\ast)+O\left(\frac{\log_q(n)}{n}\right)\right)}.$ Hence asymptotically, the bounds become equality $|\mathcal{B}(\mathbf{r},\tau n)\cap \Sigma_q^n|=q^{n\left(H_q(\kappa^\ast)+O\left(\frac{\log_q(n)}{n}\right)\right)}.$

\end{proof}


Then, we consider when the received word ${\bf r}\in \Sigma_q^{m}\setminus R_q(m).$
In the remainder of this section, the following notation is used for the received word.

{\begin{defn}\label{nonrepr}
Let $\mathbf{r}\in\Sigma_q^m\setminus R_q(m).$ Let $0<w<m$ be the Hamming weight of $\mathbf{r}, w=\wt_{\sf H}({\bf r}).$
In general, the received word $\mathbf{r}$ has the following form
\begin{equation}~\label{eq:r}
{\bf r}=\left((0)^{a_1}, x_1 , (0)^{a_2},x_2,\cdots, (0)^{a_{w}}, x_w, (0)^{a_{w+1}}\right),
\end{equation}
where $a_1,\cdots, a_{w+1}\geq 0$, $a_1+\cdots+a_{w+1}= m-w$ and $x_1,\cdots,x_w\in\Sigma_q\setminus\{0\}.$
 Furthermore, define $0\leq t\leq w$ be an integer such that there are $w+1-t$ zero runs of $\mathbf{r}.$ That is,
$t=\left|\left\{i\in\{1,\cdots,w+1\}:a_i=0\right\}\right|.$
\end{defn}

For the analysis of such words, first we define the definition of runs.

\begin{defn}\label{def:run}
Let $\mathbf{r}$ be a vector. Define $\varphi(\mathbf{r})$ be the number of runs in $\mathbf{r}$ where a run in $\mathbf{r}$ is a maximum consecutive identical symbol in $\mathbf{r}.$ For example, the number of runs in $\mathbf{r} = (0,1,1,0)$ is $3$ while the number of runs in $\mathbf{r} = (0,1,0,1)$ is $4.$
\end{defn}
}

\begin{lemma}\label{Lemmarun}
Assuming $q\geq 3,\varphi(\mathbf{r})$ can be tightly bounded by
\begin{eqnarray*}
2(w-t)+1\leq \varphi(\mathbf{r})\leq 2w-t+1.
\end{eqnarray*}
When $q=2,$ the bounds are also applicable but the upper bound is only tight when $t\leq 2.$ When $t\geq 2,$ $\varphi(\mathbf{r})$ is upper bounded tightly by $\varphi(\mathbf{r})\leq 2(w-t)+3.$
\end{lemma}
\begin{proof}
Note that when $t=0, \varphi({\bf r})=2w+1.$  Having one of these $a_i$ to be $0$ will decrease $\varphi({\bf r})$ by at least $1,$ since the run $(0)^{a_i}$ itself is removed. On the other hand, the most reduction to the number of runs that $a_i=0$ can cause is $2.$ It happens when $2\leq i\leq w$ and $x_{i-1}=x_{i}.$ Thus, $2w+1-2t \leq \varphi(\mathbf{r}) \leq 2w+1-t.$ Noting that all non-zero elements must be the same when $q=2,$ the argument above provides us with $\varphi(\mathbf{r})\leq 2(w-t)+3.$ It is easy to see that the upper bound is tight when $t\geq 2.$

To prove the bounds are tight for $q\geq 3,$ it is sufficient to construct two $\mathbf{r}$ with number of runs achieving the two bounds. Consider $a_2=a_3=\cdots=a_{t+1}=0$ and $x_1=\cdots=x_{t+1}=1.$ Then, the received word
${\bf r}= \left((0)^{a_1}, (1)^{t+1}, (0)^{a_{t+2}}, x_{t+2},\cdots, (0)^{a_w}, x_w, (0)^{a_{w+1}}\right),$
where $\varphi({\bf r}) = 2(w-t)+1$ proving the tightness of the lower bound. Note that this also proves the tightness of the lower bound for $q=2.$
Consider $a_2=\cdots=a_{t-1}=0$, $x_1=x_3=\cdots=x_{2i+1}=\cdots = 1$ and $x_2=x_4=\cdots = x_{2i}=\cdots=\alpha$ for some non-zero $\alpha\in\Sigma_q$ with $\alpha\neq 1.$ So, we have the received word ${\bf r}= \left(1,\alpha,1,\alpha,\cdots, (0)^{a_{t}}, x_{t},\cdots, (0)^{a_w}, x_w\right),$ where $\varphi({\bf r}) = 2w-t+1$ proving the tightness of the upper bound when $q\geq 3.$
\end{proof}

For our analysis in the remainder of this section, we will be using the following two results regarding the size of insertion and deletion spheres.

\begin{lemma}[see in~\cite{InsBall}]~\label{InsertionBallSize}
For any non-negative integer $n_2$ and a vector $\mathbf{s}\in\Sigma_q^{n_1},$ the size of $\mathcal{S}_{\sf I}(\mathbf{s},n_2)$ can be exactly calculated by
\begin{equation*}
|\mathcal{S}_{\sf I}(\mathbf{s},n_2)| = \sum_{i=0}^{n_2} \binom{n_1+n_2}{i}(q-1)^i.
\end{equation*}
\end{lemma}

\begin{lemma}[see in~\cite{DSH}]~\label{thmDSH}
For any non-negative integer $n_2\leq n_1$ and a vector ${\bf s}\in\Sigma_q^{n_1},$ the size of $\mathcal{S}_{\sf D}\left({\bf s},n_2\right)$ can be tightly bounded by
\begin{equation*}
\sum_{i=0}^{n_2} \binom{\varphi({\bf s})-n_2}{i}\leq\left|\mathcal{S}_{\sf D}\left({\bf s},n_2\right)\right|\leq\binom{\varphi({\bf s})+n_2-1 }{n_2}.
\end{equation*}
\end{lemma}

\begin{rmk}~\label{rmk:1}
Having the tight bounds of {$\varphi(\mathbf{r})$ in Lemma~\ref{Lemmarun} and $\mathcal{S}_{\sf D}(\mathbf{r},\gamma n)$} in Lemma~\ref{thmDSH}, these bounds result in the following tight bounds
{\small
\begin{eqnarray}\label{q>3bound}
q^{(2w-2t+1-\gamma n) H_q\left(\frac{\gamma n}{2w-2t+1-\gamma n}\right)+{\log_q(n)}} (q-1)^{-\gamma n} \leq& |\mathcal{S}_{\sf D}(\mathbf{r},\gamma n)| \leq q^{(2w-t+\gamma n) H_q\left(\frac{\gamma n}{2w-t+\gamma n}\right)+{\log_q(n)}} (q-1)^{-\gamma n}
\end{eqnarray}}
for $q\geq 3$ and
\begin{eqnarray}\label{q=2bound}
2^{(2w-2t+1-\gamma n) H_2\left(\frac{\gamma n}{2w-2t+1-\gamma n}\right)+{\log_q(n)}} \leq& |\mathcal{S}_{\sf D}(\mathbf{r},\gamma n)| \leq 2^{(2(w-t)+3+\gamma n) H_2\left(\frac{\gamma n}{2(w-t)+3+\gamma n}\right)+{\log_q(n)}}
\end{eqnarray}
when $q=2.$ The tightness here is in the sense that they are the maximum and minimum values of $|\mathcal{S}_{\sf D}(\mathbf{r},\gamma n)|$ for $\mathbf{r}\in\Sigma_q^{m}\setminus R_q(m)$.
\end{rmk}

\begin{lemma}\label{lm:2.1}
Let $\mathcal{C}\subseteq\Sigma_q^{n}$ be an insdel code and $\mathbf{r}\in\Sigma_q^{m}\setminus R_q(m)$ be a received word with the form in Definition~\ref{nonrepr}, such that $m\in[n-\kappa n, n+\gamma n], 0\leq \kappa< \frac{q-1}{q}, \gamma < q-1$ and $\kappa + \gamma = \tau \geq 0.$ Let $\gamma^* = \frac{\tau n - n + m}{2n}$ and $\kappa^* = \frac{\tau n + n-m}{2n}.$ Then, for $q\geq 3,$ the size of $\mathcal{B}(\mathbf{r},\tau n)\cap \Sigma_q^{n}$ is upper bounded by
\begin{equation*}
|\mathcal{B}(\mathbf{r},\tau n)\cap \Sigma_q^{n}|\leq q^{(2w-t+\gamma^* n) H_q\left(\frac{\gamma^* n}{2w-t+\gamma^* n}\right)-\gamma^{*}n\log_q(q-1) + n H_q(\kappa^*)+ { O(\log_q(n))}}.
\end{equation*}
When $q=2,$ the size of $\mathcal{B}(\mathbf{r},\tau n)\cap \Sigma_2^n$ is upper bounded by
\begin{equation*}
|\mathcal{B}(\mathbf{r},\tau n)\cap \Sigma_2^{n}|\leq 2^{(2(w-t)+2+\gamma^* n) H_2\left(\frac{\gamma^* n}{2(w-t)+2+\gamma^* n}\right)+ n H_2(\kappa^*)+ { O(\log_q(n))}}.
\end{equation*}
\end{lemma}\begin{proof} The following proof works for $q\geq 3.$ The proof can also be applied for $q=2$ by using Inequality~\eqref{q=2bound}. Let $\mathbf{r}$ be a received word of length $m.$ During the transmission, suppose that $\gamma n$ insertions and $\kappa n= n +\gamma n - m$ deletions occur. Thus, $n+\gamma n-\kappa n=m.$
To enumerate the elements in the ball, first we enumerate the elements in $S_{\sf D}(\mathbf{r},\gamma n).$ Then the size of $S_{\sf I}(\mathbf{c}',\kappa n)$ is calculated for $\mathbf{c}'\in  S_{\sf D}(\mathbf{r},\gamma n).$
\begin{eqnarray*}
|\mathcal{B}(\mathbf{r},\tau n)\cap \Sigma_q^{n}|&\leq& \sum_{\gamma=\max\left\{\frac{m-n}{n},0\right\}}^{ \frac{\tau n + m - n}{2n}} \binom{\varphi(\mathbf{r})+\gamma n - 1}{\gamma n}\sum_{i=0}^{\kappa n} \binom{n}{i}(q-1)^i\\
&\leq& \sum_{\gamma=\max\left\{\frac{m-n}{n},0\right\}}^{ \frac{\tau n + m - n}{2n}}  q^{(2w-t+\gamma n) H_q\left(\frac{\gamma n}{2w-t+\gamma n}\right)+{ O(\log_q(n))}}\cdot (q-1)^{-\gamma n}\sum_{i=0}^{\kappa n} q^{nH_q\left(\frac{i}{n}\right)}\\
&\leq& \sum_{\gamma=\max\left\{\frac{m-n}{n},0\right\}}^{ \frac{\tau n + m - n}{2n}} q^{(2w-t+\gamma n)H_q\left(\frac{\gamma n}{2w-t+\gamma n}\right)-\gamma n \log_q(q-1) + nH_q(\kappa)+{ O(\log_q(n))}}.
\end{eqnarray*}

Since $(2w-t+\gamma n)H_q\left(\frac{\gamma n}{2w-t+\gamma n}\right)-\gamma n \log_q(q-1)$ and $nH_q(\kappa)$ are increasing functions on $\gamma\leq \gamma^\ast$ and $\kappa\leq \kappa^\ast$ respectively, we have

\begin{equation*}
|\mathcal{B}(\mathbf{r},\tau n)\cap \Sigma_q^{n}|\leq q^{(2w-t+\gamma^* n) H_q\left(\frac{\gamma^* n}{2w-t+\gamma^* n}\right)-\gamma^{*}n\log_q(q-1) + n H_q(\kappa^*)+{ O(\log_q(n))}}.
\end{equation*}
\end{proof}

It is interesting to look at the Gilbert-Varshamov bound for the insdel codes.
The Gilbert-Varshamov bound is known as an upper bound on the list decoding radius under Hamming metric codes~\cite{VS05}, rank-metric codes~\cite{Ding2015}, cover-metric codes~\cite{Sliu2018} and symbol-pair metric codes~\cite{symbol2018}. Thus, any codes under Hamming metric, rank-metric, cover-metric or symbol-pair metric that are list decoded beyond this bound will output an exponential list size. A natural question is whether the Gilbert-Varshamov bound is also the limit to the list decoding of insdel codes.


Denote by $A_q(n,2\delta n)$  the maximum cardinality of insdel codes with minimum insdel distance $2\delta n$ in $\Sigma_q^{n}$, then we can obtain the following lower bound of the Gilbert-Varshamov bound.

\begin{prop}~\label{pro:1}
Let $0<\delta\leq\frac{q-1}{q},$ then
\begin{eqnarray*}
\lim_{n\rightarrow\infty}\frac{\log_q A_q(n,2\delta n)}{n}\geq 1-(1+\delta) H_q\left(\frac{\delta}{1+\delta}\right) + {\delta} \log_q(q-1) - H_q({\delta}).
\end{eqnarray*}

\end{prop}
\begin{proof}
Consider  $|\mathcal{B}(\mathbf{r},2\delta n-1)\cap \Sigma_q^{n}|$ with $\mathbf{r}\in\Sigma_q^n,$ we must have $\gamma = \kappa \leq \delta-\frac{1}{2}<\frac{q-1}{q}.$ For any $\mathbf{r}\in\Sigma_q^n\setminus R_q(n),$ we can simplify the upper  bound of Lemma~\ref{lm:2.1} to
\begin{equation*}
|\mathcal{B}(\mathbf{r},2\delta n-1)\cap \Sigma_q^n|\leq q^{\left(n-1+{\delta n-\frac{1}{2}}\right) H_q\left(\frac{\delta n-\frac{1}{2}}{n-1+\delta n-\frac{1}{2}}\right)-(\delta n-\frac{1}{2}) \log_q(q-1) + n H_q\left({\delta}\right)+{ O(\log_q(n))}}.
\end{equation*}
Note that this bound also applies to $q=2$ although it is not tight.
For any $\mathbf{r}_1, \mathbf{r}_2 \in R_q(n), \mathbf{r}_1\neq \mathbf{r}_2,$ we have $d(\mathbf{r}_1,\mathbf{r}_2) = 2n>2\delta n.$ So, we can have an insdel code $\mathcal{C}$ such that $R_q(n)\subseteq\mathcal{C}.$ Note that by Lemma~\ref{theorem:case1}, these repetition codewords eliminate at most $q\cdot q^{n\left(H_q(\kappa^\ast)+O\left(\frac{\log_q(n)}{n}\right)\right)}$ elements of $\Sigma_q^n$ from being in $\mathcal{C}.$  Thus, we have\begin{eqnarray*}
A_q(n,2\delta n) &\geq& q+ \frac{q^n - q^{nH_q\left({\delta}\right)+1+{ O(\log_q(n))}}}{q^{\left(n-1+{\delta n-\frac{1}{2}}\right) H_q\left(\frac{\delta n-\frac{1}{2}}{n-1+\delta n}\right) + { O(\log_q(n))}- {(\delta n-\frac{1}{2})} \log_q(q-1) + n H_q({\delta})}}\\
&\geq& \frac{q\cdot\left(q^{\left(n-1+{\delta n-\frac{1}{2}}\right) H_q\left(\frac{\delta n-\frac{1}{2}}{n-1+\delta n}\right) + { O(\log_q(n))}- {(\delta n-\frac{1}{2})} \log_q(q-1) + n H_q({\delta})}\right)+q^n}{q^{\left(n-1+{\delta n-\frac{1}{2}}\right) H_q\left(\frac{\delta n-\frac{1}{2}}{n-1+\delta n}\right) + { O(\log_q(n))}- {(\delta n-\frac{1}{2})} \log_q(q-1) + n H_q({\delta})}}\\
&\geq&\frac{q^n}{q^{\left(n-1+{\delta n-\frac{1}{2}}\right) H_q\left(\frac{\delta n-\frac{1}{2}}{n-1+\delta n}\right) + { O(\log_q(n))}- {(\delta n-\frac{1}{2})} \log_q(q-1) + n H_q({\delta})}}\\
&=& q^{n\left(1-\left(1-\frac{3}{2n} + {\delta}\right)H_q\left(\frac{\delta}{1-\frac{3}{2n}+\delta}\right)-O\left(\frac{\log_q(n)}{n}\right) +{(\delta-\frac{1}{2n}}\log_q(q-1) - H_q\left({\delta}\right)\right)}.
\end{eqnarray*}
Then, taking limit as $n$ tends to $\infty,$
\begin{eqnarray*}
\lim_{n\rightarrow\infty}\frac{\log_q A_q(n,2\delta n)}{n}\geq 1-(1+\delta) H_q\left(\frac{\delta}{1+\delta}\right) + {\delta} \log_q(q-1) - H_q({\delta}),\end{eqnarray*}
 we obtain the desired result.
\end{proof}

We also consider the value of $A_q(n,2\delta n)$ when $\frac{q-1}{q}<\delta<1.$

\begin{prop}\label{GVsparse}
Let $\frac{q-1}{q}<\delta<1,$ then $A_q(n,2\delta n)=q$ which directly implies $\lim_{n\rightarrow\infty}\frac{\log_q A_q(n,2\delta n)}{n}=0.$
\end{prop}
\begin{proof}
For $\mathbf{v}\in\Sigma_q^n,$ denote by $n_\mathbf{v}\triangleq \max_{\alpha\in\Sigma_q}\{n_\alpha(\mathbf{v})\}$ the largest occurrence of any element of $\Sigma_q$ in $\mathbf{v}.$ Furthermore, for any $\alpha\in\Sigma_q,$ denote by $V_\alpha(q,n)\triangleq\{\mathbf{v}\in\Sigma_q^n: n_\alpha(\mathbf{v})\geq \frac{n}{q}\}$ the set of all vectors over $\Sigma_q$ of length $n$ which has at least $\frac{n}{q}$ entries having value $\alpha.$ By Pigeonhole Principle, we have that for any $\mathbf{v}\in\Sigma_q^n,$ there must exist $\alpha\in\Sigma_q$ such that $n_\alpha(\mathbf{v})\geq \frac{n}{q}.$ This implies $\Sigma_q^n = \bigcup_{\alpha\in\Sigma_q}V_\alpha(q,n).$

Now note that for any $\alpha\in\Sigma_q$ and $\mathbf{u},\mathbf{v}\in V_\alpha(q,n)$ two distinct elements of $V_\alpha(q,n),$ from $\mathbf{u},$ we can simply delete all entries except $\frac{n}{q}$ occurrences of $\alpha$ (which is guaranteed by the fact that $\mathbf{u},\mathbf{v}\in V_\alpha(q,n))$  and then insert the appropriate entries to obtain $\mathbf{v}.$ So $d(\mathbf{u},\mathbf{v})\leq \frac{2n(q-1)}{q}.$ This shows that if $\mathcal{C}\subseteq\Sigma_q^n$ is an insdel code of minimum insdel distance $2\delta n>\frac{2n(q-1)}{q},$ for any $\alpha\in\Sigma_q, |\mathcal{C}\cap V_\alpha(q,n)|\leq 1.$ So by union bound, $|\mathcal{C}|\leq \sum_{\alpha\in\Sigma_q}|\mathcal{C}\cap V_\alpha(q,n)|\leq q.$ Hence $A_q(n,2\delta n)\leq q.$

To show equality, we construct an insdel code $\mathcal{C}\subseteq \Sigma_q^n$ of minimum insdel distance $2\delta n.$ Consider the following code
\begin{equation*}
\mathcal{C} =\left\{ (0)^{n-\delta n} (\alpha)^{\delta n}, \alpha\in \Sigma_q\right\}.
\end{equation*}
It is easy to see that $|\mathcal{C}|=q$ and $d(\mathcal{C})= 2\delta n.$ This shows the existence of an insdel code of minimum insdel distance $2\delta n$ over $\Sigma_q^n$ proving that $A_q(n,2\delta n)=q,$ concluding the proof.
\end{proof}
\begin{rmk}
Proposition~\ref{GVsparse} shows that for relative minimum insdel distance larger than $\frac{q-1}{q},$ there does not exist any asymptotically good code. Hence  the asymptotic behaviour of codes of relative minimum insdel distance beyond $\frac{q-1}{q}$ is not of interest and all our analysis will be done with the estimate of the bounds when the relative minimum insdel distance is at most $\frac{q-1}{q},$ which is derived in Proposition~\ref{pro:1}.
\end{rmk}

\begin{rmk}


Given any $\epsilon\in (0,1),$ when $q=2^{\Omega(1/\epsilon)},$ the lower bound in Proposition~\ref{pro:1} can be simplified to $R\geq 1-\delta-\epsilon,$ which is a lower bound of the Gilbert-Varshamov bound.
On the other hand, the insdel Gilbert-Varshamov bound is naturally upper bounded by the insdel Singleton bound,  $R\leq 1-\delta.$ This shows that the gap between the lower bound in Proposition~\ref{pro:1} and the insdel Singleton bound, where the insdel Gilbert-Varshamov bound lies, can be made arbitrarily small when $q$ is sufficiently large. This shows that when $q$ is sufficiently large, the three bounds can be made arbitrarily close to each other.




\end{rmk}

\section{List decoding of random insdel codes}\label{sec:random}
In this section, we investigate the list decodability of  random insdel codes.

\begin{theorem}\label{thm:2}
Let $q\geq 3$ and $\epsilon\in(0,1)$ be small. Then with probability at least $1-q^{-n},$ a random insdel code $\mathcal{C}\subseteq\Sigma_q^n$ of rate
\begin{equation}\label{RRand3}
R=\min_{\tiny
\begin{array}{c}
\gamma\in[0,q-1),\\
\kappa\in\left[0,\frac{q-1}{q}\right),\\
\gamma+\kappa=\tau
\end{array}
}1-(2\gamma-\kappa+1) H_q\left(\frac{\gamma}{2\gamma-\kappa+1}\right)+\gamma\log_q(q-1)-H_q(\kappa)-\epsilon
\end{equation}
is $(\tau n,O(1/\epsilon))$-list-decodable for sufficiently large $n.$
\end{theorem}
\begin{proof}
Let $L=\left\lceil\frac{1+\tau}{\epsilon}\right\rceil-1$ and $n$ be a sufficiently large positive integer. Pick an insdel code $\mC$ with size $q^{Rn}$ uniformly at random.  We calculate the probability that $\mC$ is not $(\tau n, L)$-list-decodable.

If $\mC$ is not $(\tau n, L)$-list-decodable, there exists a word ${\bf r}\in \Sigma_{q}^{m}$  for a positive integer $m\in[n-\tau n,n+\tau n]$ and a subset $\mS\subseteq \mC$ with $|\mS|=L+1$ such that $\mS\subseteq \mB({\bf r},\tau n).$

If ${\bf r}\in \Sigma_q^m\setminus R_q(m),$ by Lemma~\ref{lm:2.1}, the probability that one codeword $c\in \mathcal{C}$ is contained in $\mathcal{B}(\mathbf{r},\tau n)$ is at most $q^{(m-1+\gamma^* n)H_q\left(\frac{\gamma^* n}{m-1+\gamma^* n}\right) - \gamma^* n \log_q(q-1)+ n H_q(\kappa^*)+{ O(\log_q(n))}-n}$ where $\gamma^\ast n = \frac{\tau n-n+m}{2}$ and $\kappa^\ast n=\frac{\tau n + n-m}{2}.$ Together with Lemma~\ref{theorem:case1}, for a uniformly sampled $\mathbf{r},$ we have
\begin{eqnarray}\label{eq:1}
\nonumber \Pr[c\in \mB({\bf r},  \tau n)]&=&\frac{|\mB({\bf r}, \tau n)\cap\Sigma_q^n|}{q^{n}}\\
 \nonumber&=&\frac{q}{q^{m}}\cdot \frac{|\mB({\bf r}\in R_q(m), \tau n)\cap\Sigma_q^n|}{q^{n}}+\frac{q^{m}-q}{q^{m}}\cdot \frac{|\mB({\bf r}\in\Sigma_q^{m}\setminus R_q(m), \tau n)\cap\Sigma_q^n|}{q^{n}}\\
\nonumber &\leq& q^{nH_q(\kappa^\ast)+1-m-n+{ O(\log_q(n))}}\left(q+(q^m-q)q^{(m-1+\gamma^* n) H_q\left(\frac{\gamma^* n}{m-1+\gamma^* n}\right)-\gamma^{*}n\log_q(q-1)}\right)\\
 &\leq&{q^{(m-1+\gamma^* n) H_q\left(\frac{\gamma^* n}{m-1+\gamma^* n}\right)-\gamma^{*}n\log_q(q-1) + n H_q(\kappa^*)-n+{ O(\log_q(n))}}}.
 \end{eqnarray}

 Let $E_{{\bf r},\mS}$ be the event that all codewords in $\mS$ are contained in $\mB({\bf r}, \tau n)$.
By Equation~\eqref{eq:1}, we have
\begin{eqnarray*}
\Pr[E_{{\bf r},S}]&\leq&\left(\frac{|\mB({\bf r}, \tau n)\cap\Sigma_q^n|}{q^{n}}\right)^{L+1}\\
&\leq& {\left({q^{(m-1+\gamma^* n) H_q\left(\frac{\gamma^* n}{m-1+\gamma^* n}\right)-\gamma^{*}n\log_q(q-1) + n H_q(\kappa^*)-n+{ O(\log_q(n))}}}\right)}^{L+1}.
\end{eqnarray*}

Let $\hat{\gamma}\in[0,q-1)$ and $\hat{\kappa}\in\left[0,\frac{q-1}{q}\right)$ be the values of $\gamma$ and $\kappa$ that maximize $(2\gamma+1-\kappa)H_q\left(\frac{\gamma}{2\gamma+1-\kappa}\right)-\gamma\log_q(q-1)+H_q(\kappa).$ Taking the union bound over all choices of $m, q^m$ choices of $\mathbf{r}$ and $\mathcal{S}$ over any $(L+1)-$ subsets of $\mathcal{C},$ we have

\begin{eqnarray}\label{eq:2}
\nonumber\sum_{{\bf r},\mS}\Pr[E_{{\bf r},\mS}]&\leq& \sum_{m=n-\tau n}^{n+\tau n}q^{m}\binom{|\mC|}{L+1}{\left({q^{(m-1+\gamma^* n) H_q\left(\frac{\gamma^* n}{m-1+\gamma^* n}\right)-\gamma^{*}n\log_q(q-1) + n H_q(\kappa^*)-n+{ O(\log_q(n))}}}\right)}^{L+1}\\
\nonumber&\leq&\sum_{m=n-\tau n}^{n+\tau n} q^m q^{n(L+1)\left[R+\left(2\hat{\gamma}-\hat{\kappa}+1-\frac{1}{n}\right)H_q\left(\frac{\hat{\gamma}}{2\hat{\gamma}+1-\hat{\kappa}}\right)-\hat{\gamma}\log_q(q-1)+H_q(\hat{\kappa})-1+O\left(\frac{\log_q(n)}{n}\right)\right]}\\
\nonumber&\leq&q^{(1+\tau)n}q^{n(L+1)\left[R+\left(2\hat{\gamma}-\hat{\kappa}+1\right)H_q\left(\frac{\hat{\gamma}}{2\hat{\gamma}+1-\hat{\kappa}}\right)-\hat{\gamma}\log_q(q-1)+H_q(\hat{\kappa})-1+O\left(\frac{\log_q(n)}{n}\right)\right]}\\
\nonumber&\leq&q^{-n}.
\end{eqnarray}

\end{proof}


The bound of random insdel codes  when $q=3$ is observed in Figure~\ref{q=3maintheorems}.
In order to compare to other existing results more easily, the proof of Theorem~\ref{thm:2} can be applied to provide the list decodability of random insdel codes against $\gamma n$ insertions and $\kappa n$ deletions.

\begin{lemma}\label{thm2fixed}
Fix an alphabet size $q\geq 3$ and let $\gamma\in[0,q-1)$ and $\kappa\in\left[0,\frac{q-1}{q}\right).$  Then, for a small $\epsilon\in(0,1)$ and a sufficiently large $n,$  with probability at least $1-q^{-n},$ a random insdel code $\mathcal{C}\subseteq\Sigma_q^n$ of rate
\begin{equation}\label{RRand3fixed}
R=1-\left(2\gamma-\kappa+1\right)H_q\left(\frac{\gamma}{2\gamma-\kappa+1}\right)+\gamma\log_q(q-1)-H_q(\kappa)-\Ge
\end{equation}
is list decodable against $\gamma n$ insertions and $\kappa n$ deletions with list size $O(1/\epsilon).$
\end{lemma}

\begin{rmk}\label{rmk:compjaprandboth}
Lemma~\ref{thm2fixed} improves the list decoding radius of random insdel codes in~\cite{Japan} for any $q.$ This can be observed in Figure~\ref{compjaprandboth2fig}.
\vspace{0cm}
\begin{figure}[H]
\begin{center}
\includegraphics[width=28em]{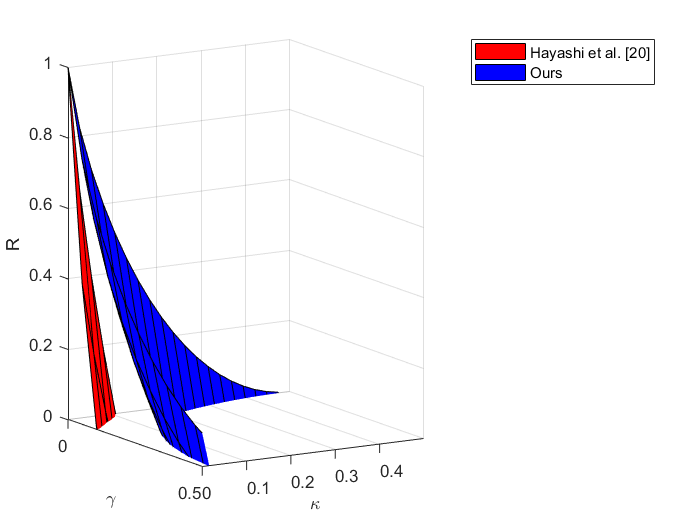}
\end{center}
\vspace{-0.5 cm}
\caption{Comparison Figure in Remark~\ref{rmk:compjaprandboth} with $q=2$}\label{compjaprandboth2fig}
\end{figure}
\end{rmk}

\begin{cor}\label{randq} Let $\epsilon\in(0,1)$ be small and suppose that $q=2^{\Omega(1/\epsilon)}.$ Then for any $\gamma\in[0,q-1)$ and $\kappa\in\left[0,\frac{q-1}{q}\right)$ and sufficiently large $n,$  with probability at least $1-q^{-n},$ a random insdel code $\mathcal{C}\subseteq\Sigma_q^n$ of rate $R=1-\kappa-\epsilon$ is list decodable against $\gamma n$ insertions and $\kappa n$ deletions with list size $O(1/\epsilon).$
 \end{cor}


\begin{rmk}\label{RandSBGV}
Let $\epsilon\in(0,1)$ be small and the alphabet $q= 2^{\Omega{(1/\epsilon)}}$. If there are $\gamma$ fraction of insertions and $\kappa$ fraction of deletions with $\gamma+\kappa=\tau$ and $\gamma>\kappa+2\epsilon,$
there exists a $(\tau n,O(1/\epsilon))$-list-decodable insdel code with list decoding radius $\tau n$ beyond the minimum insdel distance $d$.
Generally, the list decoding radius of codes cannot break the minimum distance barrier. This is true for codes in Hamming metric, rank-metric, symbol-pair and cover-metric. Interestingly, under insdel distance, some insdel codes can be list decoded beyond the minimum insdel distance with polynomial list size.

\end{rmk}

Analogue to Theorem~\ref{thm:2} and Lemma~\ref{thm2fixed}, similar results can be derived for the binary case.

\begin{cor}[Binary Case]\label{thm:2q=2}
Let $\epsilon\in(0,1)$ be small.
Let $\hat{\gamma}\in[0,1)$ and $\hat{\kappa}\in\left[0,\frac{1}{2}\right)$
satisfy
$(\hat{\gamma},\hat{\kappa})=\arg \max_{\tiny\begin{array}{c}\gamma\in[0,1),\\\kappa\in[0,1/2)\end{array}}\left\{(2\theta(\gamma,\kappa)+\gamma)H_2\left(\frac{\gamma}{2\theta(\gamma,\kappa)+\gamma}\right)+H_2(\kappa)-(1+\gamma-\kappa)+(1+\gamma-\kappa)H_2\left(\frac{2\theta(\gamma,\kappa)}{1+\gamma-\kappa}\right)\right\}$
where $\theta(\gamma,\kappa)\triangleq\frac{1+2\gamma-\kappa}{8}+\frac{\sqrt{(1+\gamma-\kappa)^2+10\gamma(1+\gamma-\kappa)+\gamma^2}}{8}.$ We also denote $\hat{\theta}=\theta(\hat{\gamma},\hat{\kappa}).$ Then with probability at least $1-2^{-n},$  a random binary insdel code $\mC\subseteq\Sigma_2^n$ of rate
\begin{equation}\label{RRand2}
R=1-(2\hat{\theta} +\hat{\gamma})H_2\left(\frac{\hat{\gamma}}{2\hat{\theta}+\hat{\gamma}}\right)-H_2(\hat{\kappa})+(1+\hat{\gamma}-\hat{\kappa})-(1+\hat{\gamma}-\hat{\kappa})H_2\left(\frac{2\hat{\theta}}{1+\hat{\gamma}-\hat{\kappa}}\right)-\Ge
\end{equation}
is $(\tau n, O(1/\epsilon))-$ list-decodable for sufficiently large $n.$
\end{cor}

The visualization of the bound provided in Corollary~\ref{thm:2q=2} can be found in Figure~\ref{q=2maintheorems} in Section~\ref{ConsSection}. 


\begin{cor}[Binary Case]\label{thm:2q=2fixed}
Let $\gamma\in[0,1)$ and $\kappa\in\left[0,\frac{1}{2}\right)$ and a small $\epsilon\in(0,1).$ Then for a sufficiently large $n,$ with probability at least $1-2^{-n},$ a random binary insdel code $\mathcal{C}\subseteq\Sigma_2^n$ of rate
\begin{equation}\label{RRand2fixed}
R=1-(2\theta +\gamma)H_2\left(\frac{\gamma}{2\theta+\gamma}\right)-H_2(\kappa)+(1+\gamma-\kappa)-(1+\gamma-\kappa)H_2\left(\frac{2\theta}{1+\gamma-\kappa}\right)-\Ge
\end{equation}
is list decodable against $\gamma n$ insertions and $\kappa n$ deletions with list size $O(1/\epsilon)$ where $\theta=\frac{1+2\gamma-\kappa}{8}+\frac{\sqrt{(1+\gamma-\kappa)^2+10\gamma(1+\gamma-\kappa)+\gamma^2}}{8}.$
\end{cor}

By Lemma~\ref{thm2fixed} and Corollary~\ref{thm:2q=2fixed}, we have the following corollaries when only insertions (or deletions) occur.

\begin{cor}[Insertions only]\label{cor:whpins}
For every small $\Ge\in(0,1)$ and $0\leq\gamma<q-1$ fraction of insertions, with probability at least $1-q^{-n}$, a random code $\mC\subseteq\Sigma_q^{n}$ of rate
\begin{eqnarray*}
R=\left\{
\begin{array}{cl}
1-\left(1+2\gamma\right)H_q\left(\frac{\gamma}{1+2\gamma}\right)+\gamma\log_q(q-1)-\Ge &\textrm{if~} q\geq 3~\textrm{and}\\
1-(2\theta+\gamma)H_2\left(\frac{\gamma}{2\theta+\gamma}\right) + (1+\gamma)-(1+\gamma)H_2\left(\frac{2\theta}{1+\gamma}\right)-\epsilon &\textrm{if~}q=2
\end{array}
\right.
\end{eqnarray*}
is list decodable against any $\gamma n $ insertions for all sufficiently large $n$ with list size $L=O(1/\epsilon)$ and $\theta=\theta(\gamma,0),$ where $\theta(\gamma,\kappa)$ is defined in Corollary~\ref{thm:2q=2fixed}.
\end{cor}

\begin{rmk}\label{rmk:whpins}
We compare Corollary~\ref{cor:whpins} with the result in~\cite[Theorem $1.7$]{HSS} in three cases; $q=2,q\geq 3$ and $q= 2^{\Omega(1/\epsilon)}.$ Firstly, let $q=2.$ Fixing the value of $\epsilon>0$ and the list size guaranteed by the two bounds, plotting the two curves provides that the rate in Corollary~\ref{cor:whpins} is better than the rate in \cite[Theorem $1.7$]{HSS}. This can be observed from Figure~\ref{comprandins2fig}. When $q\geq 3,$ it can be shown that with fixed $\epsilon$ and the list size $\mathcal{L}=\frac{\gamma+1}{\epsilon}-1,$ the rate provided in Corollary ~\ref{cor:whpins} is worse than the rate provided in \cite[Theorem $1.7$]{HSS}. Lastly, when $q= 2^{\Omega(1/\epsilon)},$ fixing the values of $\epsilon$ and $R=1-\epsilon,$ the list size required in Corollary~\ref{cor:whpins} is $\left\lceil\frac{\gamma+1}{\epsilon}\right\rceil-1$ while~\cite[Theorem $1.7$]{HSS} requires $L>\frac{\gamma+1}{\epsilon}-1.$ So the two list size requirements differ by at most $1,$ which happens when $\gamma+1$ is an integer multiple of $\epsilon.$
\begin{figure}[H]
\begin{center}
\includegraphics[scale=0.62]{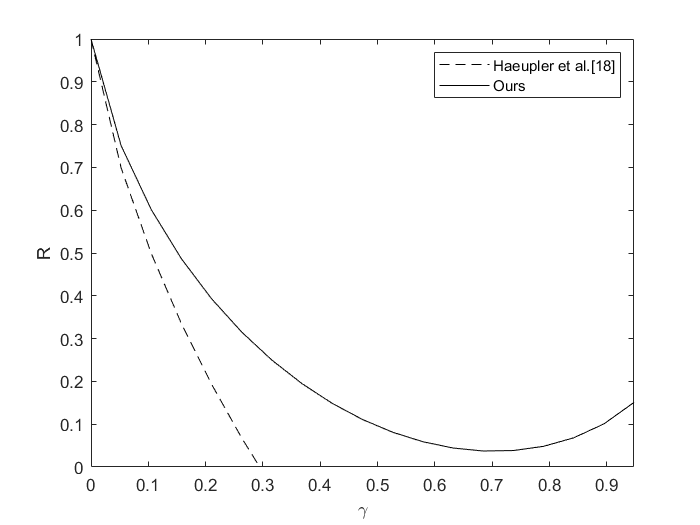}
\end{center}
\vspace{-0.5 cm}
\caption{Comparison Figure of Remark~\ref{rmk:whpins} for insertions only with $q=2$ and $\epsilon=0.0001.$}\label{comprandins2fig}
\end{figure}
\end{rmk}


%
%
%
%
%

\begin{cor}[Deletions only]\label{cor:whpdel}
For every small $\Ge\in(0,1)$ and $0\leq\kappa<1$
with probability at least $1-q^{-n}$, a random code $\mC\subseteq\Sigma_q^{n}$ of rate
$
R=1-H_q(\kappa)-\Ge
$
is list decodable against any $\kappa$ fraction of deletions for all sufficiently large $n$ with list size $\mathcal{L}=O(1/\epsilon).$
\end{cor}

\begin{rmk}\label{rmk:whpdel}
The comparison between Corollary~\ref{cor:whpdel} and \cite[Theorem $1.6$]{HSS} reveals the following result. Fix $\epsilon>0$ and list size $\mL=\left\lceil\frac{1+\gamma}{\epsilon}\right\rceil-1$, the rates of the random list decodable code $\mathcal{C}$ reaches the same value $R=1-H_q(\kappa)-\epsilon.$ Considering the list decodability of random binary deletion codes, the same analysis reveals that the list decoding radius in Corollary~\ref{cor:whpdel} is same as that in~\cite[Theorem $26$]{VGCW}. These observations are illustrated in Figure~\ref{randcomp}.

\begin{figure}[H]
\hspace{-3 em}
\subfloat[Comparison with Theorem $1.6$ in~\cite{HSS} for $q=256,\epsilon=0.0001$ \label{compsudranddel256e-3fig}]{%
       \includegraphics[width=0.6\textwidth]{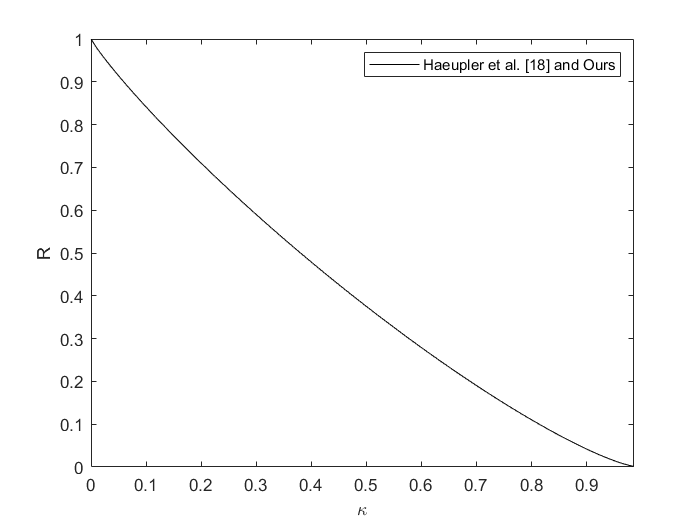}
     }
     \hspace{-2.8em}
     \subfloat[Comparison with Theorem 26 in~\cite{VGCW}\label{compgurranddel2fig} for $q=2,\epsilon=0.0001$]{%
       \includegraphics[width=0.6\textwidth]{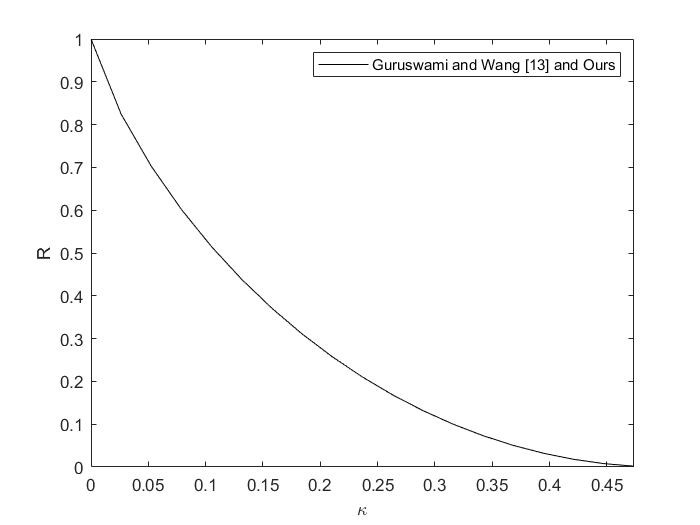}
     }
\caption{Comparison Figures in Remark~\ref{rmk:whpdel} for deletions only.}\label{randcomp}
\end{figure}
\end{rmk}

%
%

Reducing the sample space from arbitrary insdel codes to arbitrary $\Sigma_q$-linear insdel code in the above theorems, we have the following results.
\begin{theorem}\label{thm:randlinstrong}
let $q\geq 3$ and $\epsilon\in(0,1)$ be small. Let $\hat{\gamma}\in [0,q-1)$ and $\hat{\kappa}\in\left[0,\frac{q-1}{q}\right)$ satisfy $(\hat{\gamma},\hat{\kappa})=\arg \max_{\tiny\begin{array}{c}\gamma\in[0,q-1),\\\kappa\in[0,(q-1)/q)\end{array}}(2\gamma-\kappa+1)H_q\left(\frac{\gamma}{2\gamma-\kappa+1}\right)-\gamma\log_q(q-1)+H_q(\kappa).$ Then with probability at least $1-q^{-n},$ a random $\Sigma_q$-linear insdel code $\mC\subseteq\Sigma_q^n$ of rate
\begin{eqnarray*}
R=1-\left(2\hat{\gamma}-\hat{\kappa}+1\right)H_q\left(\frac{\hat{\gamma}}{2\hat{\gamma}-\hat{\kappa}+1}\right)+\hat{\gamma}\log_q(q-1)-H_q(\hat{\kappa})-\Ge
\end{eqnarray*}
is $\left(\tau n, \exp(O(1/\epsilon))\right)$-list-decodable for sufficiently large $n.$
\end{theorem}

\begin{proof}
Let $L=q^{\left\lceil\frac{\tau+1}{\epsilon}\right\rceil}-1$ and $n$ be a sufficiently large integer. Then $\log_q(L+1)=\left\lceil\frac{\tau+1}{\epsilon}\right\rceil$ and $L=\exp(O(\frac{1}{\epsilon})).$

Pick $Rn$ $\Sigma_q$-linearly independent words uniformly at random from $\Sigma_q^{n}$. The $\Sigma_q$-linear insdel code $\mC$ spanned by these words has rate $R$.
If  $\mC$ is not $(\tau n, L)$-list-decodable, then there exists a word ${\bf r}\in \mathbb{F}_{q}^{m}$ for a positive integer $m\in[n-\tau n,n+\tau n]$ and a subset $\mS\subseteq \mC$ with $|\mS|=L+1$ such that $\mS\subseteq \mB({\bf r}, \tau n).$ There are at least $L'=\log_q (L+1)={\lceil {\frac{\tau+1}{\Ge}}\rceil}$ codewords in $\mS$ which are $\Sigma_q$-linearly independent. Let $\mS'$ be the $\Sigma_q$-linear span of these $L'$ codewords, thus $\mS'\subseteq \mS.$ Then, $\Pr[E_{{\bf r},\mS}]\leq \Pr[E_{{\bf r},\mS'}]$ and by Lemmas~\ref{theorem:case1} and~\ref{lm:2.1}, for $\gamma^\ast=\frac{\tau n-n+m}{2n}\in[0,q-1)$ and $\kappa^\ast=\frac{\tau n+n-m}{2}\in\left[0,\frac{q-1}{q}\right),$
\begin{eqnarray*}
\Pr[E_{{\bf r},\mS'}]
&\leq&\left(\frac{|\mB({\bf r}, \tau n)|}{q^{n}}\right)^{L'}\\
&=&\left(\frac{q}{q^{m}}\cdot \frac{|\mB({\bf r}\in R_q(m), \tau n)\cap\Sigma_q^n|}{q^{n}}+\frac{q^{m}-q}{q^{m}}\cdot \frac{|\mB({\bf r}\in\Sigma_q^{m}\setminus R_q(m), \tau n)\cap\Sigma_q^n|}{q^{n}}\right)^{L'}\\
&\leq&{\left({q^{(m-1+\gamma^* n) H_q\left(\frac{\gamma^* n}{m-1+\gamma^* n}\right)-\gamma^{*}n\log_q(q-1) + n H_q(\kappa^*)-n+{ O(\log_q(n))}}}\right)}^{L'}\\
&=&q^{nL'\left(\left(2\gamma^\ast+1-\kappa^\ast-\frac{1}{n}\right)H_q\left(\frac{\gamma^\ast}{2\gamma^\ast+1-\kappa^\ast-\frac{1}{n}}\right)-\gamma^\ast\log_q(q-1)+H_q(\kappa^\ast)-1+O({\log_q(n)}/n)\right)}\\
&\leq&q^{nL'\left((2\hat{\gamma}+1-\hat{\kappa})H_q\left(\frac{\hat{\gamma}}{2\hat{\gamma}+1-\hat{\kappa}}\right)-\hat{\gamma}\log_q(q-1)+H_q(\hat{\kappa})-1+O({\log_q(n)}/n)\right)}.
\end{eqnarray*}

Taking the union bound over all choices of $m,q^m$ choices for $\mathbf{r}$ and any $L~ \Sigma_q$-linearly independent words from $\mathcal{C},$ we can derive the following probability.

\begin{eqnarray}\label{eq:3}
\nonumber\sum_{{\bf r},S}\Pr[E_{{\bf r},\mS}]&\leq& \sum_{m=n-\tau n}^{n+\tau n}q^{m}\binom{|\mC|}{L'}q^{nL'\left((2\hat{\gamma}+1-\hat{\kappa})H_q\left(\frac{\hat{\gamma}}{2\hat{\gamma}+1-\hat{\kappa}}\right)-\hat{\gamma}\log_q(q-1)+H_q(\hat{\kappa})-1+O({\log_q(n)}/n)\right)}\\
\nonumber &\leq& q^{(\tau+1)n}|\mC|^{L'}{q^{nL'\left((2\hat{\gamma}+1-\hat{\kappa})H_q\left(\frac{\hat{\gamma}}{2\hat{\gamma}+1-\hat{\kappa}}\right)-\hat{\gamma}\log_q(q-1)+H_q(\hat{\kappa})-1+O({\log_q(n)}/n)\right)}}\\
\nonumber&\leq&q^{nL'\left(\frac{\tau+1}{L'}+R+(2\hat{\gamma}+1-\hat{\kappa})H_q\left(\frac{\hat{\gamma}}{2\hat{\gamma}+1-\hat{\kappa}}\right)-\hat{\gamma}\log_q(q-1)+H_q(\hat{\kappa})-1+O({\log_q(n)}/n)\right)}\\
\nonumber&\leq&q^{-n}.
\end{eqnarray}

\end{proof}

Similar proof technique can be applied to provide the following result.

\begin{lemma}\label{lem:randlinweak}
Fix an alphabet size $q\geq 3$ and let $\gamma\in[0,q-1)$ and $\kappa\in\left[0,\frac{q-1}{q}\right).$ Then for a small $\epsilon\in (0,1)$ and a sufficiently large $n,$ with probability at least $1-q^{-n},$ a random $\Sigma_q$-linear insdel code $\mathcal{C}\subseteq\Sigma_q^n$ of rate
\begin{eqnarray*}
R=1-\left(2\gamma-\kappa+1\right)H_q\left(\frac{\gamma}{2\gamma-\kappa+1}\right)+\gamma\log_q(q-1)-H_q(\kappa)-\Ge
\end{eqnarray*}
is list decodable against $\gamma n$ insertions and $\kappa n$ deletions with list size $\exp(O(1/\epsilon)).$
\end{lemma}

Analogue to Theorem~\ref{thm:randlinstrong} and Lemma~\ref{lem:randlinweak}, similar results can be derived for the binary case.

\begin{cor}[Binary Case]
Let $\epsilon\in(0,1)$ be small. For any $\gamma\in[0,1)$ and $\kappa\in\left[0,\frac{1}{2}\right),$ let $\theta(\gamma,\kappa),\hat{\gamma},\hat{\kappa}$ and $\hat{\theta}$ be as defined in Corollary~\ref{thm:2q=2}. Then with probability at least $1-2^{-n},$ a random binary linear insdel code $\mathcal{C}\subseteq \Sigma_2^n$ of rate

\begin{eqnarray*}
R=1-(2\hat{\theta} -\hat{\gamma})H_2\left(\frac{\hat{\gamma}}{2\hat{\theta}+\hat{\gamma}}\right)-H_2(\hat{\kappa})+(1+\hat{\gamma}-\hat{\kappa})-(1+\hat{\gamma}-\hat{\kappa})H_2\left(\frac{2\hat{\theta}}{1+\hat{\gamma}-\hat{\kappa}}\right)-\Ge
\end{eqnarray*}
is $\left(\tau n, \exp(O(1/\epsilon))\right)$-list-decodable for all sufficiently large $n.$
\end{cor}

\begin{cor}[Binary Case]
Let $\gamma\in[0,1)$ and $\kappa\in\left[0,\frac{1}{2}\right)$ and a small $\epsilon\in(0,1).$ Then for a sufficiently large $n,$ with probability at least $1-2^{-n},$ a random binary linear insdel code $\mathcal{C}\subseteq\Sigma_2^n$ of rate

\begin{eqnarray*}
R=1-(2\theta -\gamma)H_2\left(\frac{\gamma}{2\theta+\gamma}\right)-H_2(\kappa)+(1+\gamma-\kappa)-(1+\gamma-\kappa)H_2\left(\frac{2\theta}{1+\gamma-\kappa}\right)-\Ge
\end{eqnarray*}
is list decodable against $\gamma n$ insertions and $\kappa n$ deletions with list size $\exp(O(1/\epsilon))$ where $\theta=\frac{1+2\gamma-\kappa}{8}+\frac{\sqrt{(1+\gamma-\kappa)^2+10\gamma(1+\gamma-\kappa)+\gamma^2}}{8}.$
\end{cor}

\begin{cor}
Let $\epsilon\in(0,1)$ be small and suppose that $q=2^{\Omega(1/\epsilon)}.$ Then for any $\gamma\in[0,q-1)$ and $\kappa\in\left[0,\frac{q-1}{q}\right)$ and sufficiently large $n,$  with probability at least $1-q^{-n},$ a $\Sigma_q$-linear random insdel code $\mC\subseteq\Sigma_q^n$ of rate $R=1-\kappa-\epsilon$ is list decodable against $\gamma n$ insertions and $\kappa n$ deletions with list size $\exp(O(1/\epsilon)).$
\end{cor}


\section{Explicit insdel codes with list decoding algorithm}\label{ConsSection}
In this section, we provide an explicit construction of a family of insdel codes that has an efficient decoding algorithm. Similar to \cite{VR16}, \cite{VGCW} and \cite{Japan}, the construction is done by concatenation method and indexing scheme. In \cite{Japan}, they constructed a family of insdel codes with list decoding radius up to the Johnson-type bound and designed its efficient algorithm. The construction done by Haeupler, Shahrasbi and Sudan in \cite{HSS} provided a family of list decodable insdel codes when the alphabet size is sufficiently large. The construction considered there has list decoding radius achieving the limit of list decoding radius they have derived. This paper focuses on the construction of a family of explicit list decodable insdel codes for smaller alphabet size, even when $q=2.$


Denote our concatenated code by $\mathcal{C}_{conc},$ with inner code $\mathcal{C}_{in}$ and outer code by $\mathcal{C}_{out}.$ The outer code $\mathcal{C}_{out}$ is chosen to be a $p$-ary code of length $N$ and rate $R_{out}$ that is $(\alpha_{out},\ell_{out},\mathcal{L}_{out})$-list-recoverable. The inner code $\mathcal{C}_{in}$ is chosen to be a random $q$-ary code of length $n$ and rate $R_{in}.$ By Theorems~\ref{thm:2} and~\ref{thm:2q=2}, $\mathcal{C}_{in}$ is $(\tau_{in} n, O(1/\epsilon_{in}))$-list-decodable with rate $R_{in}.$  To obtain the codewords in $\mathbf{c}\in\mathcal{C}_{conc}$ from the outer codeword $\mathbf{c}_{out}=(c_1,\cdots, c_N)\in \mathcal{C}_{out},$ index each $c_i$ by $i\pmod{\epsilon_{cont}N}+1$ for some values $\epsilon_{cont}$ that will be determined later and encode $(i\pmod{\epsilon_{cont} N}+1,c_i)$ with the encoding function $\varphi_{in}:[\epsilon_{cont} N]\times\Sigma_{p}\to\Sigma_q^n$ of $\mathcal{C}_{in}, \mathbf{c}=(\varphi_{in}(1,c_1),\cdots,\varphi_{in}((N\pmod{\epsilon_{cont}N}+1,c_N)).$

Let $\mathbf{c}\in\mathcal{C}_{conc}\subseteq \Sigma_q^{nN}$ be the sent codeword and $M\in[\max\{0,nN-\tau nN\},nN+\tau nN]$ be the length of the received word $\mathbf{r}=(r_1,\cdots, r_M) $ such that $d(\mathbf{c},\mathbf{r})\leq \tau nN.$ Denote $\mathbf{c}=(\mathbf{v}_1,\cdots, \mathbf{v}_N)$ where $\mathbf{v}_i\in \mathcal{C}_{in}$ is the $i$-th block of $\mathbf{c}$ and $\mathbf{r}=(\mathbf{w}_1,\cdots, \mathbf{w}_N)$ such that $\mathbf{w}_i$ is obtained from $\mathbf{v}_i.$ Denote by $\tau_i n =d(\mathbf{v}_i,\mathbf{w}_i).$ Then $\sum_{i=1}^N \tau_i n\leq \tau nN.$ Figure~\ref{vtow} provides an illustration of the notations.

\vspace{-0.2cm}
\begin{figure}[H]
\begin{center}
\includegraphics[scale=0.87]{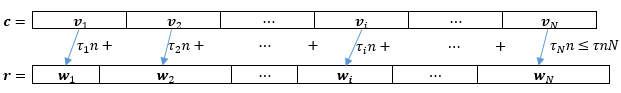}
\end{center}
\vspace{-0.75 cm}
\caption{Partition of the sent codeword $\mathbf{c}$ and the received word $\mathbf{r}$}\label{vtow}
\end{figure}

\subsection{Construction of subsequences of the received word}\label{Sisubsection}
Let $0<\tau^\ast<\tau_{in}$ and $\hat{\tau}\triangleq\tau_{in}-\tau^\ast.$ Define the following set of subsequences of $\mathbf{r}:$
\begin{equation}\label{Sconst}
\mathcal{S}\triangleq\left\{\begin{array}{c}
(r_{1+\Phi},\cdots, r_{\Phi+\Lambda}):\Phi=\lambda\hat{\tau}n,\Lambda=\mu\hat{\tau}n,\lambda,\mu\in\mathbb{Z}, \\
0\leq\lambda\leq 1+\frac{\frac{M}{n}-\max(0,1-\tau^\ast)}{\hat{\tau}},\max\left(0,\frac{1-\tau^\ast}{\hat{\tau}}\right)\leq \mu\leq 1+\frac{1+\tau^\ast}{\hat{\tau}}
\end{array}
\right\}.
\end{equation}

\begin{lemma}\label{wtoslemma}
Take a subsequence $\mathbf{w}$ of $\mathbf{r}$ with $\mathbf{w}=(r_{1+sp},\cdots, r_{sp+len}).$ Then there exists a subsequence $\mathfrak{s}=(r_{1+\Phi},\cdots, r_{\Phi+\Lambda})\in\mathcal{S}$ such that both $\Phi$ and $\Lambda$ are integer multiples of ${\hat{\tau}} n$ and $d(\mathfrak{s},\mathbf{w})\leq \hat{\tau} n.$
\end{lemma}
\begin{proof}
Let $\varphi_{sp}$ and $\varphi_{len}$ be non-negative integers and $sp_{\hat{\tau}}$ and $len_{\hat{\tau}}$ be non-negative real numbers such that $sp=\varphi_{sp} {\hat{\tau}} n + sp_{{\hat{\tau}}}$ and $len = \varphi_{len} {\hat{\tau}} n + len_{\hat{\tau}}$ where $0\leq sp_{\hat{\tau}},len_{\hat{\tau}} < \hat{\tau} n.$ Utilizing these notations, $\mathbf{w}$ can be rewritten as $\mathbf{w} = (r_{1+\varphi_{sp}{\hat{\tau}} n+ sp_{\hat{\tau}}},\cdots, r_{(\varphi_{sp}+\varphi_{len}){\hat{\tau}} n+ (sp_{\hat{\tau}} + len_{\hat{\tau}})}).$ By assumption, we obtain $0\leq sp_{\hat{\tau}} + len_{\hat{\tau}} < 2{\hat{\tau}} n$ and divide into two cases to discuss
\vspace{-0.5em}
\begin{enumerate}
\item When $0\leq sp_{\hat{\tau}} + len_{\hat{\tau}} < {\hat{\tau}} n,$ then  $(\varphi_{sp}+\varphi_{len}){\hat{\tau}} n+ (sp_{\hat{\tau}} + len_{\hat{\tau}})< (\varphi_{sp}+\varphi_{len}+1){\hat{\tau}} n.$ Set $\Phi = \varphi_{sp}{\hat{\tau}} n$ and $\Lambda = (\varphi_{len}+1){\hat{\tau}} n.$ This implies
$\mathfrak{s} = (r_{\varphi_{sp}{\hat{\tau}} n},\cdots, r_{(\varphi_{sp}+\varphi_{len}+1){\hat{\tau}} n-1}).$
In this case $\mathbf{w}$ is a subsequence of $\mathfrak{s}.$ Hence $d(\mathbf{w},\mathfrak{s})$ is at most the difference in the lengths, $d(\mathbf{w},\mathfrak{s}) = (\varphi_{len}+1){\hat{\tau}} n- (\varphi_{len}{\hat{\tau}} n+ len_{\hat{\tau}}) = {\hat{\tau}} n - len_{\hat{\tau}} \leq {\hat{\tau}} n. $
\item When ${\hat{\tau}} n\leq sp_{\hat{\tau}} + len_{\hat{\tau}} < 2{\hat{\tau}} n,$ then $(\varphi_{sp}+\varphi_{len}){\hat{\tau}} n+ (sp_{\hat{\tau}} + len_{\hat{\tau}})\geq (\varphi_{sp}+\varphi_{len}+1){\hat{\tau}} n.$ Set $\Phi = (\varphi_{sp}+1){\hat{\tau}} n$ and $\Lambda = \varphi_{len}{\hat{\tau}} n.$ This implies
$\mathfrak{s} = (r_{(\varphi_{sp}+1){\hat{\tau}} n},\cdots, r_{(\varphi_{sp}+\varphi_{len}+1){\hat{\tau}} n-1}).$
In this case $\mathfrak{s}$ is a subsequence of $\mathbf{w}.$ Hence $d(\mathbf{w},\mathfrak{s})$ is at most the difference in the lengths, $d(\mathbf{w},\mathfrak{s}) = (\varphi_{len}{\hat{\tau}} n + len_{\hat{\tau}})- (\varphi_{len}{\hat{\tau}} n) = len_{\hat{\tau}} < {\hat{\tau}} n. $
\end{enumerate}
\vspace{-1em}
\end{proof}

Note that $|\mathcal{S}|\leq \left(\frac{(1+\tau)N-\max(0,1-\tau^\ast)}{\tau_{in}-\tau^\ast}\right)\frac{\min(2\tau^\ast,1+\tau^\ast)}{\tau_{in}-\tau^\ast}=O(N).$ Using the list decodability of $\mathcal{C}_{in},$ each $\mathbf{s}=(r_{1+sp},\cdots,r_{sp+len})\in\mathcal{S}$ can be list decoded to a list of size $O\left(\frac{1}{\epsilon_{in}}\right).$ Recall that the domain of the inner encoding has an indexing scheme with indices from $1$ to $\epsilon_{cont} N.$ So any element of the resulting list is in the form $(i, c)\in [\epsilon_{cont} N]\times \Sigma_{p^m}.$ Based on the index introduced, $c$ is then a possible value of the entries of the outer codeword in the indices that is $i\pmod {\epsilon_{cont} N}+1.$ Now we consider whether $c$ is a possible value of all entries of the outer codeword in the indices that is $i\pmod{\epsilon_{cont}N}+1.$

Let $j\in[N]$ such that $j-1\equiv i\pmod{\epsilon_{cont} N}.$ Then there exists a non-negative integer $j_N$ such that $j=1+i+j_N\epsilon_{cont} N.$ Fix the notations $\mathbf{v}_j^{(L)}=(\mathbf{v}_1,\cdots, \mathbf{v}_{j-1}), \mathfrak{s}_j^{(L)}=(r_1,\cdots, r_{sp}), \mathbf{v}_j^{(R)}=(\mathbf{v}_{j+1},\cdots, \mathbf{v}_{N})$ and $\mathfrak{s}_j^{(R)}=(r_{1+sp+len},\cdots, r_M).$ Lastly, denote by $\tau_j n = d(\mathbf{v}_j,\mathfrak{s}),\tau_i^{(L)}n=d(\mathbf{v}_j^{(L)},\mathfrak{s}_j^{(L)})$ and $\tau_j^{(R)}n=d(\mathbf{v}_j^{(R)},\mathfrak{s}_j^{(R)}).$ Then $\tau_j^{(L)}n + \tau_j n + \tau_j^{(R)}n = d(\mathbf{c},\mathbf{r})\leq \tau nN.$ This leads to the following requirements that $(sp,len)$ needs to satisfy
\begin{enumerate}
\item $\tau_{in} n\geq |n-len|,$
\item $\tau nN- |n-len|\geq \tau_j^{(L)} n\geq |sp-(j-1)n|,$
\item $\tau nN - |n-len|- |sp+(j-1)n|\geq \tau_j^{(R)} n\geq |(N-j)n-sp-\ell|$ and
\item $0\leq sp\leq M-len$ since we have $1\leq 1+sp\leq sp+len\leq M.$
\end{enumerate}

For all possible values of $M,$ these requirements are met if and only if $\max(n-\tau_{in} n, n-|M-Nn|)\leq len\leq \min(n+\tau_{in} n,n+|M-Nn|)$ and $\max\left(0,(j-1)n- \frac{(1+\tau)nN-M}{2}+\max(n-len,0)\right)\leq sp\leq \min\left(M-len,(j-1)n + \frac{M-(1-\tau)nN}{2}+\min(n-len,0)\right).$
Here $\mathcal{W}_i$ can be constructed depending on the values of $M$ and $\tau^\ast.$ In general, this requirements are equivalent to
$\max\{0,n-\tau_{in} n\} \leq len\leq n+\tau_{in} n$ and $(j-1)n-\frac{(1+\tau)nN-M}{2}
+\max(n-len,0)\leq sp\leq (j-1)n+\frac{M-(1-\tau)nN}{2}+\min(n-len,0).$ Setting $sp=\lambda\hat{\tau} n, len = \mu \hat{\tau} n, j = 1+i+j_N \epsilon_{cont}N, $ for some non-negative integers $\lambda$ and $\mu$ the requirements become
\begin{equation}\label{mureq}
\frac{\max(0,1-\tau_{in})}{\hat{\tau}}\leq \mu \leq \frac{1+\tau_{in}}{\hat{\tau}}
\end{equation}
and
\begin{equation}\label{lambdareq}
\frac{i+j_N\epsilon_{cont}nN - \frac{(1+\tau)nN-M}{2}}{\hat{\tau}n}+\max\left(\frac{1}{\hat{\tau}}-\mu,0\right)\leq \lambda\leq \frac{i+j_N\epsilon_{cont}nN +\frac{M-(1-\tau)nN}{2}}{\hat{\tau}n}+\min\left(\frac{1}{\hat{\tau}}-\mu,0\right).
\end{equation}
Fixing $i,\lambda$ and $\mu,$ This gives us the following requirement on the value of $j_N.$
\begin{equation}\label{jNreq}
\max\left(0,\frac{\frac{(1-\tau)nN-M}{2}+\lambda\hat{\tau}n - \min(n-\mu\hat{\tau}n,0)-i}{\epsilon_{cont}nN}\right)\leq j_N\leq\frac{\frac{(1+\tau)nN-M}{2}+\lambda\hat{\tau}n-\max(n-\mu\hat{\tau}n,0)-i}{\epsilon_{cont}nN}.
\end{equation}

So fixing $i,\lambda$ and $\mu,$ the number of possible $j_N$ is at most $\frac{\tau}{\epsilon_{cont}}.$ In total, the sum of the sizes of the positional lists $\ell_{out}$ is at most $O\left(\frac{\tau N}{\epsilon_{in}\epsilon_{cont}}\right).$

}

\subsection{Construction and list decoding algorithm}
\begin{theorem}\label{ConsThm}
Let $\epsilon_{conc},\epsilon_{cont},\epsilon_{out},R_{out},\epsilon_{in},R_{in}\in(0,1)$ be positive real numbers. Furthermore, take $m=\frac{\ell_{out}}{N\zeta^2}.$ Let $\mathcal{C}_{out}\subseteq \F_{N^{2m}}^N$ be a code of rate $R_{out}$ and $\left(\alpha_{out}:=R_{out}+\epsilon_{out},\ell_{out}, L_{out}:=N^\frac{2\ell_{out}}{N\zeta}\right)-$list-recoverable for some $\ell_{out}=O(N)$ and $\zeta$ satisfying $\frac{(\alpha_{out}-\zeta)(1-\zeta)}{1-\frac{N\zeta}{\ell_{out}}}<\alpha_{out}-\epsilon_{out}$ with list-recovering complexity $T(N).$ Set $\mathcal{C}_{in}\subseteq \Sigma_q^n$ to be an insdel code by Theorem~\ref{thm:2} (Corollary~\ref{thm:2q=2} for binary case) depending on the value of $q$ that has rate $R_{in}$ and is $\left(\tau_{in} n,O\left(\frac{1}{\epsilon_{in}}\right)\right)-$list-decodable where the relation between $R_{in}$ and $\tau_{in}$ is determined by Equation~\eqref{RRand2} if $q=2$ and Equation~\eqref{RRand3} otherwise. Lastly, choose $0<\tau^\ast<\tau_{in}$ such that $\tau_{in}-\frac{\epsilon_{conc}}{1-\alpha_{out}}\leq \tau^\ast.$

Using the concatenation method described above, $\mathcal{C}_{conc}$ is a list decodable insdel code of rate
$R_{conc}=R_{out}R_{in}-\epsilon$
and it is $\left(((1-\alpha_{out})\tau_{in}-\epsilon_{conc})nN,N^{O\left(\frac{1}{\epsilon_{in}\zeta}\right)}\right)-$list-decodable for some small $\epsilon.$ Furthermore, the list decoding algorithm has complexity $\mathrm{poly}(N)+T(N).$
\end{theorem}
\begin{proof}
First, we discuss the rate of the concatenated code. Based on the concatenation method described above, $R_{in}= \frac{\log_q\epsilon_{cont} +(1+2m)\log_q N}{n}$ and $R_{conc}=\frac{R_{out}}{1+\frac{1}{2m}}\cdot\left(R_{in}-\frac{\log_q \epsilon_{cont}}{n}\right).$ For any $\epsilon>0, \epsilon_{cont}$ and $\epsilon_{out}$ can be chosen such that $R_{conc}\geq R_{out}R_{in}-\epsilon.$

Then we discuss the decoding algorithm. The idea is to list decode sufficiently many ``windows'' from $\mathcal{S}$ we described in Subsection~\ref{Sisubsection}. Apply the list-recovering algorithm for the outer code to the resulting lists for each entries $A_1,\cdots, A_N.$ The full algorithm can be found in Algorithm~\ref{LDA}.
\vspace{-0.28cm}
\begin{algorithm}[H]
\caption{List Decoding Algorithm for $\mathcal{C}_{conc}$}
\begin{algorithmic}[1]\label{LDA}
\REQUIRE Received word $\mathbf{r}\in \Sigma_q^M, \max\{0,(1-\tau)nN\}\leq M\leq (1+\tau)nN.$
\STATE Set $A_1,\cdots,A_N\leftarrow\emptyset;$
\STATE Construct $\mathcal{S}$ as discussed in Subsection~\ref{Sisubsection};
\FOR{$\mathbf{s}\in\mathcal{S}$}
	\FOR{$(i,\alpha)\in [\epsilon_{cont}N]\times \Sigma_{N^{2m}}$}
	\STATE Calculate $\mathbf{c}_\alpha=\varphi_{in}(i,\alpha);$
	\IF{$d(\mathbf{c}_\alpha,\mathbf{s})\leq \tau_{in} n$}
	\FOR{$j=0,\cdots, \frac{1}{\epsilon_{cont}}-1$}
		\IF{$j$ satisfies the requirement in~\eqref{jNreq}}
			\STATE $A_{j\epsilon_{cont}N+i+1}\leftarrow A_{j\epsilon_{cont}N+i+1}\cup\{\alpha\};$
		\ENDIF
	\ENDFOR
	\ENDIF
	\ENDFOR
\ENDFOR
\STATE Apply list-recovering algorithm for $\mathcal{C}_{out}$ with positional lists $A_1,\cdots, A_N$ to get $L_{out}\subseteq \mathcal{C}_{out}$ and apply $\varphi_{in}$ for each codewords to get $L\subseteq \mathcal{C}_{conc}$ of the same size;
\RETURN $L$;
\end{algorithmic}
\end{algorithm}

Next we discuss the correctness of the decoding algorithm. An index $i$ is said to be ``good'' if $d(\mathbf{v}_i,\mathbf{w}_i)=\tau_i n\leq \tau^\ast n$ and ``bad'' otherwise. Since $\sum_{i=1}^N \tau_i \leq \tau N,$ if there are $h$ ``bad'' indices, $\tau N> h\tau^\ast.$ This implies that there are at most $\frac{\tau}{\tau^\ast}N$ ``bad'' indices. By the property of $\mathcal{S}_i, \mathbf{v}_i\in A_i$ for at least $\left(1-\frac{\tau}{\tau^\ast}\right)N\geq \alpha_{out} N$ indices. As discussed in Subsection~\ref{Sisubsection}, $\sum_{i=1}^N|A_i|\leq |\mathcal{S}| \ell_{in}=O\left(\frac{N}{\epsilon_{in}}\right)= \ell_{out}.$ Together with the list-recoverability of $\mathcal{C}_{out},$ inputting $A_1,\cdots, A_N$ to the list-recovering algorithm of $\mathcal{C}_{out}$ yields a list $L_{out}$ of size at most $N^\frac{2\ell_{out}}{N\zeta}$ of codewords of $\mathcal{C}_{out}.$ Applying $\varphi_{in}$ to each codewords in $L_{out}$ gives a list $L$ of codewords of $\mathcal{C}_{conc}$ of the same size. As discussed above, $\mathbf{v}_i\in A_i$ for at least $\alpha_{out} N$ indices. So the original sent codeword $\mathbf{c}\in L,$ proving the correctness of the decoding algorithm.

We consider the decoding complexity. Construction of $\mathcal{S}$ has complexity $O(N).$ Next, for each $\mathbf{s}\in\mathcal{S}$ and $(i,\alpha)\in[\epsilon_{cont}N]\times \Sigma_{N^{2m}},$ calculation of $d(\mathbf{c}_\alpha,\mathbf{s})$ requires finding the longest common subsequence of two strings of length $n.$ This takes $O(n^2).$ So if $T_{in}=O(\epsilon_{cont}N^{1+2m}),$ is the encoding time of the inner code, the complexity of the comparison is $O(n^2|\mathcal{S}|\epsilon_{cont}N^{1+2m})=O\left(n^2N^{2+\frac{2}{\epsilon_{in}\zeta^2}}\right).$ Next, by the list decodability of $\mathcal{C}_{in},$ we have at most $\frac{|\mathcal{S}|}{\epsilon_{in}}$ of these comparisons resulting in distance of at most $\tau_{in} n.$ Hence the assignment step of $\alpha$ to the $O\left(\frac{\tau}{\epsilon_{cont}}\right)$ positional lists has complexity $O\left(\frac{\tau N}{\epsilon_{in}\epsilon_{cont}}\right).$
Lastly, encoding each elements of $L_{out}$ takes $O\left(N^{1+\frac{1}{\epsilon_{in}\zeta}\frac{2}{\epsilon_{in}\zeta^2}}\right).$ Since every step has complexity $\mathrm{poly}(N),$ the total complexity is $\mathrm{poly}(N)+T(N).$

\end{proof}

For our final construction,$~\mathcal{C}_{out}$ is chosen from the family of list-recoverable $p$-ary codes of length $N$ and rate $R_{out}$ that can be derived from \cite[Theorem $10$]{GFRS}. In this construction, instead of making the codes to be over any $\Sigma_p,$ it is required that $\Sigma_p$ is a finite field of $p$ elements, which is denoted by $\F_p.$  This result can be transformed to a construction of list-recoverable code as can be observed in Theorem~\ref{LRCthm}.

\begin{lemma}[Adapted from {\cite[Theorem $10$]{GFRS}}]\label{LRCthm}
For $R_{out},\epsilon_{out}>0,$ a sufficiently large $N,\ell_{out}=O(N),$ there exists $\zeta>0, m=\frac{\ell_{out}}{N\zeta^2}$ and a prime power $p=O(N^2)$ such that a folded Reed-Solomon code $\mathcal{C}_{out}\subseteq \F_{p^m}^N$ of rate $R_{out}$ is $\left(\alpha_{out}:=R_{out}+\epsilon_{out},\ell_{out}, L_{out}:=p^{\frac{\ell}{N\zeta}}\right)$-list-recoverable. Here $\zeta$ is chosen to be small enough such that $\frac{(\alpha_{out}-\zeta)(1-\zeta)}{1-\frac{N\zeta}{\ell_{out}}}<\alpha_{out}-\epsilon_{out}.$ Furthermore, the basis of the list $\mL_{out}$ can be recovered in complexity $O((mN\log(p))^2)$ and $\mL_{out}$ can be recovered with time complexity $\exp(\ell_{out}/(N\zeta)).$
\end{lemma}

The construction in Theorem~\ref{ConsThm} provides the following family of insdel codes over small alphabet size that are list decodable up to a Zyablov-type bound.

\begin{theorem}[Zyablov-type bound]~\label{ExThm}
For every prime power $q,$ real numbers $0<R,\epsilon<1$ and sufficiently large $N,$ there exists a family of list decodable insdel codes of rate $R,$ length $N$ and is $\left(\tau N, N^{O\left(\frac{1}{\epsilon}\right)}\right)$-list-decodable where $\displaystyle\tau = \max_{0<R_{out},R_{in}<1} (1-R_{out})f^{-1}(R_{in})-\epsilon$ with the restriction $R_{in}R_{out}=R.$
The function $R_{in}=f(\tau_{in})$ is defined as Equation~\eqref{RRand2} if $q=2$ and Equation~\eqref{RRand3} otherwise. Lastly, this family of insdel codes can be list-decoded in $\mathrm{poly}(N)$ time.
\end{theorem}

The visualization of the Zyablov-type bound and its relation with our previous rates in
Theorem~\ref{thm:2} and Corollary~\ref{thm:2q=2} when $q=2$ and $q=3$ can be observed in Figures~\ref{q=2maintheorems} and ~\ref{q=3maintheorems} respectively.


\begin{figure}[H]
     \hspace{-2.8em}
\subfloat[Comparison Figure for $q=2$\label{q=2maintheorems}]{%
       \includegraphics[width=0.6\textwidth]{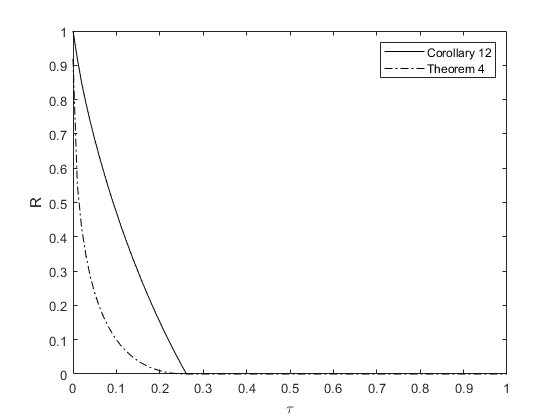}
     }
     \hspace{-2.8em}
     \subfloat[Comparison Figure for $q=3$\label{q=3maintheorems}]{%
       \includegraphics[width=0.6\textwidth]{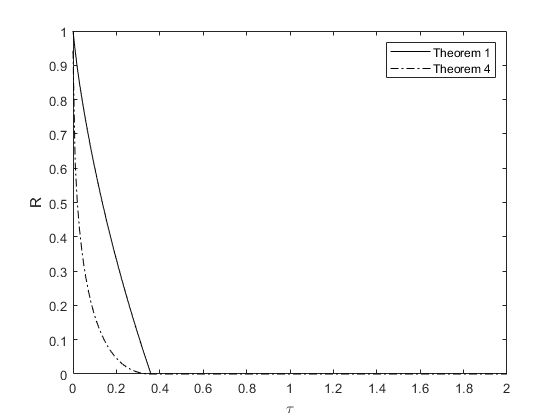}
     }
\caption{Figures of Theorem~\ref{thm:2}, Corollary~\ref{thm:2q=2} and Theorem~\ref{ExThm}.}
\end{figure}

Finally, we provide some comparisons between our construction and some existing constructions of insdel codes for small values of $q.$
In order to compare with the construction given in ~\cite{Japan}, an explicit construction of insdel codes can be done via Corollary~\ref{thm:2q=2} to obtain the Zyablov-type bound which relates the rate $R$ and the fractions of insertion errors $\gamma$ and deletion errors $\kappa$ that it needs to list decode.
\begin{lemma}\label{ExThmfixed}
For every prime power $q,$ real numbers $0<R,\epsilon<1$ and sufficiently large $N,$ there exists a family of list decodable insdel codes of rate $R,$ length $N$ and is list decodable against $\gamma N$ insertions and $\kappa N$ deletions with list size $N^{O\left(\frac{1}{\epsilon}\right)}$ where $\displaystyle \gamma=\max_{0<R_{out},R_{in}<1}(1-R_{out})\gamma_{in}-\epsilon, \kappa=\max_{0<R_{out},R_{in}<1}(1-R_{out})\kappa_{in}-\epsilon$ with the restriction $R_{in}R_{out}=R.$ Here for any given $R_{in}, \gamma_{in}$ and $\kappa_{in}$ are defined such that $f(\gamma_{in},\kappa_{in})=R_{in}$ which is defined in Equation~\eqref{RRand2fixed} if $q=2$ and Equation~\eqref{RRand3fixed} otherwise. Lastly, this family of insdel codes can be list-decoded in $\mathrm{poly}(N)$ time.
\end{lemma}

\begin{rmk}\label{JapConstComp}

The list decoding radius in Lemma~\ref{ExThmfixed} is beyond the Johnson-type bound, which improves the explicit construction of insdel codes designed by Hayashi and Yasunaga~\cite{Japan} for small alphabet size even in the binary case.
The improvement when $q=2$ can be observed in Figure~\ref{ConstvsJap2fig}.

\begin{figure}[H]
\begin{center}
\vspace{-0.4cm}
\includegraphics[scale=0.75]{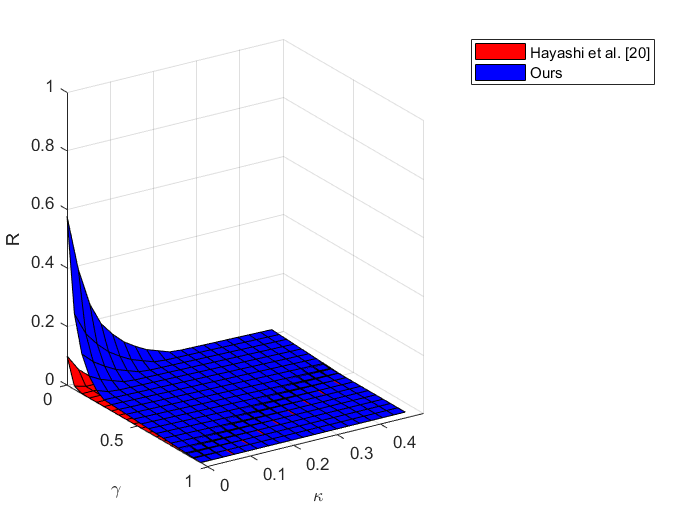}
\end{center}
\vspace{-0.5 cm}
\caption{Comparison Figure of Remark~\ref{JapConstComp} for $q=2$}\label{ConstvsJap2fig}
\end{figure}
\end{rmk}

We showed the comparison with the binary deletion codes constructed in~\cite{VGCW}.

\begin{rmk}
Guruswami and Wang~\cite{VGCW} provided an explicit construction of binary deletion codes with list decoding radius $(\frac{1}{2}-\epsilon)nN$ and polynomial list size. For deletions only, our construction has a larger range of list decoding radius $\tau\in\left(0, \frac{1}{2}\right)$ with polynomial list size.
\end{rmk}

\section*{Acknowledgment}
We would like to express our gratitude to Bernhard Haeupler and Amir Shahrasbi for the valuable discussion in improving our manuscript and noting the mistake in the previous version (v5) regarding the analysis on the upper bound of list decodability of insdel codes. As has been pointed out in~\cite{HS20}, analysis for the upper bound of the insdel ball size cannot be directly applied in the lower bound since there is a large intersection between insdel balls of various centres that needs to be considered independently. Because of this, all results based on these claims is removed from our manuscript in this new version. Although a general analysis of such intersection has been done in~\cite{Lev01}, direct application of the bound given to our case does not provide an interesting lower bound for the insdel ball size. In order to obtain an interesting lower bound for the insdel ball size (and hence an interesting upper bound on the list decodability of insdel code), a more careful analysis needs to be done and it will be our future research direction. A more careful comparison with the results presented~\cite{Lev01} is also a future research direction that is interesting to provide a more complete picture on list decodability of insdel codes.





\end{document}